\documentclass[11pt]{article}
\usepackage{amsmath,amssymb,graphicx}
\usepackage{hyperref}
\usepackage{multicol,xcolor,longtable}
\definecolor{darkred}{rgb}{0.65,0.15,0}
\hypersetup{pdfborder={0 0 0},colorlinks=true,urlcolor=darkred,citecolor=blue,linkcolor=darkred,linktocpage=true}

\usepackage{cite}

\usepackage{tikz-cd}
\usepackage{tikz}
\usetikzlibrary{arrows}
\usetikzlibrary{decorations.markings}
\usepackage{pgfplots}
\pgfplotsset{compat=1.12}

\usepackage{amsmath}
\usepackage{amsthm}
\usepackage{amssymb}
\usepackage{mathrsfs}
\usepackage[Symbol]{upgreek}
\usepackage{dsfont}
\usepackage[vcentermath]{youngtab}
\usepackage{textcomp}

\usepackage{latexsym}
\usepackage{graphicx}

\usepackage{calc}
\makeatletter
\newlength\@SizeOfCirc%
\newcommand{\CricArrowRight}[1]{%
    \setlength{\@SizeOfCirc}{\maxof{\widthof{#1}}{\heightof{#1}}}%
    \tikz [x=1.0ex,y=1.0ex,line width=.15ex, draw=black]%
        \draw [->,anchor=center]%
            node (0,0) {#1}%
            (0,1.2\@SizeOfCirc) arc (-240:85:1.2\@SizeOfCirc);%
}%
\makeatother

\numberwithin{equation}{section}

\newcommand{\sump}{\sideset{}{'}\sum}

\newcommand{\doi}[2]{\href{http://dx.doi.org/#2}{#1}}

\newcommand{\eprint}[1]{{\href{http://arxiv.org/abs/#1}{\texttt{[#1}]}}}
\newcommand{\eprintN}[1]{{\href{http://arxiv.org/abs/#1}{\texttt{#1 [hep-th]}}}}

\newcommand{\cW}{{\mathcal{W}}}
\newcommand{\cR}{{\mathcal{R}}}

\newcommand{\mf}[1]{{\mathfrak{#1}}}

\def\DJo{$\;$\kern-.4em \hbox{D\kern-.8em\raise.15ex\hbox{--}\kern.35em okovi\'c}}

\newtheorem{lemma}{Lemma}

\newcommand{\be}{\begin{equation}}
\newcommand{\ee}{\end{equation}}
\def \t {\tau}
\def \cI {\mathcal{I}}
\def \cF {\mathcal{F}}
\def \bt {\bar{\tau}}
\def \G {\Gamma}
\def \Z {\mathbb{Z}}
\def \cD {\mathfrak{D}}
\def \cV {\mathcal{V}}

\def \Laplace {\Delta}
\def \w {\omega}
\newcommand{\bp}{\begin{tikzpicture}}
\newcommand{\ep}{\end{tikzpicture}}
\newcommand{\greens}[3]{
\draw [black, thick, fill = white] (#1 - 0.2,#2 -0.2) rectangle (#1 + 0.2,#2 + 0.2);
\node [black] at (#1,#2) {$#3$};}

\newcommand{\indices}[4]{
\node [black] at (#1,#2) {\scalebox{#4}{\mbox{\Huge $#3$}}};}

\newcommand{\three}[3]{
\resizebox{.05\textwidth}{!}{
\begin{tikzpicture}
\draw [line width=1.5mm] (-2,-1.7) --  (-2.4,-1.7) -- (-2.4,1.8) -- (-2,1.8) ;
\draw [line width=1.5mm] (2,-1.7) --  (2.4,-1.7) -- (2.4,1.8) -- (2,1.8) ;
\draw [line width=1.5mm, gray, opacity = 0.3] (-2,0) .. controls (0,2) .. (2,0);
\draw [line width=1.5mm, gray, opacity = 0.3] (-2,0) .. controls (0,0) .. (2,0);
\draw [line width=1.5mm, gray, opacity = 0.3] (-2,0) .. controls (0,-2) .. (2,0);
\indices{0}{1.3}{#1}{1.5}
\indices{0}{0}{#2}{1.5}
\indices{0}{-1.3}{#3}{1.5}
\end{tikzpicture}}}

\newcommand{\greensB}[3]{
\node [black] at (#1,#2) {\scalebox{1.5}{\mbox{\Huge $#3$}}};}

\newcommand{\five}[5]{
\resizebox{.05\textwidth}{!}{
\begin{tikzpicture}
\draw[line width=1.5mm, gray, opacity = 0.3]  (-2,-2) -- (2,-2);
\draw [line width=1.5mm] (-2.8,-2.2) --  (-3.2,-2.2) -- (-3.2,1.8) -- (-2.8,1.8) ;
\draw [line width=1.5mm] (2.8,-2.2) --  (3.2,-2.2) -- (3.2,1.8) -- (2.8,1.8) ;
\draw [line width=1.5mm, gray, opacity = 0.3] (-2,-2) .. controls (-2.8,0) .. (0,2);
\draw [line width=1.5mm, gray, opacity = 0.3] (-2,-2) .. controls (-1.2,0) .. (0,2);
\draw [line width=1.5mm, gray, opacity = 0.3](2,-2) .. controls (2.8,0) .. (0,2);
\draw [line width=1.5mm, gray, opacity = 0.3](2,-2) .. controls (1.2,0) .. (0,2);
\indices{-2.4}{0}{#1}{2}
\indices{-1}{0}{#2}{2}
\indices{1}{0}{#3}{2}
\indices{2.4}{0}{#4}{2}
\indices{0}{-2}{#5}{2}
\end{tikzpicture}}}

\newcommand{\four}[4]{
\resizebox{.05\textwidth}{!}{
\begin{tikzpicture}
\draw [line width=1.5mm] (-2.8,-2.2) --  (-3.2,-2.2) -- (-3.2,1.8) -- (-2.8,1.8) ;
\draw [line width=1.5mm] (2.8,-2.2) --  (3.2,-2.2) -- (3.2,1.8) -- (2.8,1.8) ;
\draw [line width=1.5mm, gray, opacity = 0.3] (0,-2) .. controls (-2.8,0) .. (0,2);
\draw [line width=1.5mm, gray, opacity = 0.3] (0,-2) .. controls (-1.2,0) .. (0,2);
\draw [line width=1.5mm, gray, opacity = 0.3](0,-2) .. controls (2.8,0) .. (0,2);
\draw [line width=1.5mm, gray, opacity = 0.3](0,-2) .. controls (1.2,0) .. (0,2);
\indices{-2.4}{0}{#1}{2}
\indices{-0.9}{0}{#2}{2}
\indices{0.9}{0}{#3}{2}
\indices{2.4}{0}{#4}{2}
\end{tikzpicture}}}

\newcommand{\six}[6]{
\resizebox{.05\textwidth}{!}{
\begin{tikzpicture}
\draw [line width=1.5 mm] (-2,-1.4) --  (-2.4,-1.4) -- (-2.4,2.4) -- (-2,2.4) ;
\draw [line width=1.5 mm] (2,-1.4) --  (2.4,-1.4) -- (2.4,2.4) -- (2,2.4) ;
\draw [line width=1.5mm, gray, opacity = 0.3] (2,-1) -- (-2,-1) -- (0,2) -- (2,-1) -- (0,0) -- (0,2) ;
\draw [line width=1.5mm, gray, opacity = 0.3] (-2,-1) -- (0,0);

\greensB{-1}{0.5}{#1}
\greensB{0}{1}{#2}
\greensB{1}{0.5}{#3}
\greensB{1}{-0.5}{#4}
\greensB{-1}{-0.5}{#5}
\greensB{0}{-1}{#6}
\end{tikzpicture}}}

\newcommand{\point}[2]{
\filldraw[black] (#1,#2) circle (2pt);
}

\newcommand{\lineB}[3]{
\node [black] at (#1,#2) {\scalebox{1.6}{\mbox{\Huge $#3$}}};}

\newcommand{\vertB}[3]{
\node [black] at (#1,#2) {\scalebox{1.2}{\mbox{\Huge $#3$}}};}

\newcommand{\merc}[9]{
  \def\temps{#1}%
  \def\tempt{#2}%
  \def\tempp{#3}%
  \def\tempq{#4}%
  \def\tempr{#5}%
  \def\tempw{#6}%
  \def\tempvIa{#7}%
  \def\tempvIb{#8}%
  \def\tempvIc{#9}%
  \mercexp
}

\newcommand{\mercexp}[9]{
  \def\tempvIIa{#1}%
  \def\tempvIIb{#2}%
  \def\tempvIIc{#3}%
  \def\tempvIIIa{#4}%
  \def\tempvIIIb{#5}%
  \def\tempvIIIc{#6}%
  \def\tempvIVa{#7}%
  \def\tempvIVb{#8}%
  \def\tempvIVc{#9}%
\raisebox{-0.5\height}{\resizebox{.15\textwidth}{!}{
\begin{tikzpicture}
\def\boxM{4.7}
\draw [white] (-\boxM,\boxM) -- (\boxM,\boxM) -- (\boxM,-\boxM) -- (-\boxM,-\boxM) -- (-\boxM,\boxM) ;
\draw  [line width=1mm, gray, opacity = 0.3] (0,0) circle (4);

\draw [line width=1mm, gray, opacity = 0.3] (0,0) -- (0,4);
\draw [line width=1mm, gray, opacity = 0.3] (0,0) -- (3.461,-2);
\draw [line width=1mm, gray, opacity = 0.3] (0,0) -- (-3.461,-2);

\lineB{-3.5}{1.8}{\temps}
\lineB{0}{2}{\tempt}
\lineB{3.4}{1.8}{\tempp}
\lineB{1.9}{-1.2}{\tempq}
\lineB{-2}{-1.3}{\tempw}
\lineB{0}{-4}{\tempr}

\vertB{-.7}{4}{\tempvIa}
\vertB{.7}{4}{\tempvIb}
\vertB{0}{3.3}{\tempvIc}

\vertB{3.7}{-1.2}{\tempvIIa}
\vertB{3.1}{-2.7}{\tempvIIb}
\vertB{3}{-1.7}{\tempvIIc}

\vertB{-3.8}{-1.2}{\tempvIIIa}
\vertB{-3}{-1.7}{\tempvIIIb}
\vertB{-3.1}{-2.7}{\tempvIIIc}

\vertB{0}{0.6}{\tempvIVa}
\vertB{0.6}{-0.3}{\tempvIVb}
\vertB{-0.6}{-0.3}{\tempvIVc}
\end{tikzpicture}}}}

\newcommand{\nn}{\nonumber}

\newcommand{\ints}{\mathbb{Z}}
\newcommand{\reals}{\mathbb{R}}

\newcommand{\pa}{\partial}
\newcommand{\pab}{\overline{\partial}}

\begin{document}

\hypersetup{pageanchor=false}
\thispagestyle{empty}

\mbox{ }
\vspace{20mm}

\begin{center}
{\LARGE \bf Tetrahedral modular graph functions}\\[10mm]

\vspace{8mm}
\normalsize
{\large Axel Kleinschmidt${}^{1,2}$ and Valentin Verschinin${}^1$}

\vspace{10mm}
${}^1${\it Max-Planck-Institut f\"{u}r Gravitationsphysik (Albert-Einstein-Institut)\\
Am M\"{u}hlenberg 1, DE-14476 Potsdam, Germany}
\vskip 1 em
${}^2${\it International Solvay Institutes\\
ULB-Campus Plaine CP231, BE-1050 Brussels, Belgium}

\vspace{20mm}

\hrule

\vspace{10mm}

\begin{tabular}{p{12cm}}
{\small
The low-energy expansion of one-loop amplitudes in type II string theory generates a series of world-sheet integrals whose integrands can be represented by world-sheet Feynman diagrams. These integrands are modular invariant and understanding the structure of the action of the modular Laplacian on them is important for determining their contribution to string scattering amplitudes. In this paper we study a particular infinite family of such integrands associated with three-loop scalar vacuum diagrams of tetrahedral topology and find closed forms for the action of the Laplacian. We analyse the possible eigenvalues and degeneracies of the Laplace operator by group- and representation-theoretic means. 
}
\end{tabular}
\vspace{7mm}
\hrule
\end{center}

\newpage
\hypersetup{pageanchor=true}
\setcounter{page}{1}

\setcounter{tocdepth}{2}
\tableofcontents

\vspace{5mm}
\hrule
\vspace{5mm}

\section{Introduction and summary}

Scattering amplitudes are central for understanding the structure of superstring theory. Their low-energy ($\alpha'$) expansion exhibits many deep mathematical structures. For instance, the four-graviton scattering amplitude of type II string theory in $10-d$ space-time dimensions is expected to be invariant under the string U-duality group $E_{d+1}(\mathbb{Z})$~\cite{Hull:1994ys} order by order in $\alpha'$~\cite{Sen:1994fa,Green:1997tv}. This can be used together with supersymmetry to determine the lowest order derivative corrections of the form $D^{2k}R^4$ arising from the four-graviton scattering amplitude as exact (generalized) automorphic forms of the moduli~\cite{Green:1997tv,Green:1997as,Kiritsis:1997em,Pioline:1998mn,Green:1998by,Obers:1999um,Green:1999pu,Green:2005ba,Basu:2007ru,Basu:2007ck,Green:2010wi,Pioline:2010kb,Green:2010kv,Basu:2011he,Green:2011vz,Fleig:2012xa,Green:2014yxa,Bossard:2014lra,Basu:2014hsa,Bossard:2014aea,Pioline:2015yea,Bossard:2015uga,Fleig:2015vky}. As these automorphic forms are invariant under U-duality, they contain information about all orders of string perturbation theory and also non-perturbative effects. However, most results here are restricted to the four-graviton amplitude in type II in various dimensions and BPS-protected couplings associated with small automorphic representations. The automorphic forms have led to predictions of higher-genus string amplitudes and certain non-renormalisation theorems that have been confirmed by direct calculations~\cite{Berkovits:2004px,Berkovits:2009aw,Gomez:2013sla,DHoker:2013eea,DHoker:2014gfa}.

In a different direction, superstring amplitudes have been evaluated for many particles at low orders in string perturbation theory and the systematics of their $\alpha'$-expansion has been studied. At string tree level and for the scattering of $N$ open or closed strings, there are amazing systematics associated with the theory of (single-valued) multiple polylogarithms and (single-valued) multiple zeta values~\cite{Brown:2009qja,Stieberger:2006bh,Stieberger:2006te,Stieberger:2009rr,Mafra:2011nv,Mafra:2011nw,Schlotterer:2012ny,Broedel:2013aza,Broedel:2013tta,Brown:2014,Stieberger:2013wea,Stieberger:2014hba,Mafra:2016mcc}. At string one-loop order, the corresponding generalisation appears to be that of (single-valued) elliptic polylogarithms and (single-valued) elliptic multiple zeta values that is currently under construction~\cite{Mafra:2012kh,Broedel:2014vla,Broedel:2015hia,Zerbini:2015rss,DHoker:2015qmf}. `Single-valued' here indicates a certain projection on the set of multiple zeta values that has to be applied in the closed superstring case~\cite{Schnetz:2013hqa,Brown:2014}. Other references on the relations of loop integrals to multiple zeta values include~\cite{Moch:2001zr,Bloch:2013tra,Adams:2015ydq} and for other work on the modular structure of string one-loop amplitudes see for example~\cite{Angelantonj:2012gw}.

In the present paper, we are interested in functions that arise in (or are related to) the $\alpha'$-expansion of closed superstring one-loop amplitudes. A one-loop amplitude is given by an integral over the modulus $\tau$ of the world-sheet torus where the integrand is a modular $SL(2,\mathbb{Z})$-invariant function that is determined by world-sheet conformal field theory. The integrand depends on $\alpha'$ and therefore the $\alpha'$-expansion of the one-loop amplitude can be studied from an $\alpha'$-expansion of the integrand. The separation into analytic and non-analytic terms in $\alpha'$ can be effectively implemented by studying the behaviour of the integrand near the boundary of the torus moduli space (cutting off the $\tau$ integration on the $SL(2,\mathbb{Z})$ fundamental domain). 

This separation and the structure of this expansion was studied in~\cite{Green:1999pv,Green:2008uj} where a formalism was developed that represented the integrand at a given $\alpha'$-order by a Feynman diagram of the world-sheet conformal field theory. This has led to the study of the structure and systematics of such world-sheet Feynman diagrams and the associated integrands in their own right~\cite{DHoker:2015foa,DHoker:2015zfa,Basu:2015ayg,Basu:2016fpd,DHoker:2016mwo,Basu:2016xrt,Basu:2016kli,DHoker:2016quv}. Understanding the structure of the integrand is necessary for finding the integrated value that is the actual contribution to the scattering amplitude. We note that for string amplitudes with more than four external states, the integrands are not necessarily described in terms of scalar propagators only but there can also be derivatives of propagators appearing in the world-sheet Feynman diagrams~\cite{Richards:2008jg,Green:2013bza}. In~\cite{Basu:2016mmk}, it was shown that in the five-graviton case one can remove these derivatives up to order $D^{10}R^4$. However, it is not known whether this is possible to arbitrary derivative order or for more general amplitudes. Restricting to standard scalar Feynman diagrams will therefore perhaps not capture all possible contributions to string scattering. Nevertheless, the scalar Feynman diagrams exhibit already a rich mathematical structure that is worthwhile to investigate.

The integrand functions determined by the scalar world-sheet Feynman diagrams are now called \textit{modular graph functions}~\cite{DHoker:2015qmf} and several cases have been studied in great detail. For world-sheet Feynman diagrams with one and two loops, the complete structure of the connected Feynman diagrams in terms of their behaviour under the modular Laplacian has been worked out~\cite{DHoker:2015foa} and this has led to many interesting and unexpected identities among these modular graph functions~\cite{DHoker:2015zfa,DHoker:2016mwo,Basu:2016kli} that partially mirror identities of multiple polylogarithms~\cite{DHoker:2015qmf}. Beyond this complete treatment of one and two loops, some special cases of higher loop integrand functions have been analysed and some of them have been integrated~\cite{DHoker:2015foa,DHoker:2015zfa,Basu:2015ayg,Basu:2016fpd,Basu:2016xrt}. One of the main tools in the study of these functions are the modular invariant differential equations that they satisfy. These are typically inhomogeneous Laplace equations that sometimes admit an explicit integration with boundary conditions from degeneration limits of the toroidal world-sheet. 

In this paper, we will study an infinite family of modular graph functions at three-loop order on the world-sheet. We restrict to  tetrahedral Feynman diagrams but allow for an arbitrary number of vertices along the edges of the tetrahedron. In graphical notation, the functions we are interested in are associated with Feynman diagrams of the form
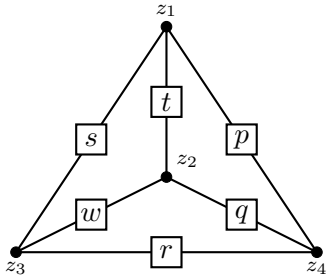
\begin{align}
\bp
\point{2}{-1}
\point{-2}{-1}
\point{0}{2}
\point{0}{0}
\draw[thick]  (2,-1) -- (-2,-1) -- (0,2) -- (2,-1) -- (0,0) -- (0,2) ;
\draw[thick]   (-2,-1) -- (0,0);
\node [above] at (0,2) {$\scriptstyle z_{1}$};
\node [above right] at (0,0) {$\scriptstyle z_{2}$};
\node [below] at (-2,-1) {$\scriptstyle z_{3}$};
\node [below] at (2,-1) {$\scriptstyle z_{4}$};
\greens{-1}{0.5}{s}
\greens{0}{1}{t}
\greens{1}{0.5}{p}
\greens{1}{-0.5}{q}
\greens{-1}{-0.5}{w}
\greens{0}{-1}{r}
\ep\nn
\end{align}
where the labels on the edges indicate the number of consecutive scalar propagators along the edge, meaning that the corresponding propagator is raised to the power given by the label. The simplest instance of such a modular graph function, corresponding to the case $s=t=p=q=w=r=1$, was studied in~\cite{Basu:2015ayg} and its contribution to the $D^{12}R^4$ derivative correction was determined by using the inhomogeneous Laplace equation satisfied by the integrand associated with this diagram. In general, we will refer to the modular graph functions associated with the above diagram as tetrahedral modular graph functions. We will call $s+t+p+q+w+r$ the \textit{weight} of the modular graph function.

We shall show in this paper that the family of tetrahedral modular graph functions satisfies an inhomogeneous Laplace equation where the right-hand side contains `simpler' modular graph functions when the spectrum is diagonalised. This is in complete parallel with~\cite{DHoker:2015foa} where at two loops the right-hand sides were quadratic polynomials in non-holomorphic Eisenstein series. Our results contain those of~\cite{Basu:2015ayg} mentioned above as a special case and we employ heavily graphical methods similar to those of~\cite{Basu:2016kli}. The tetrahedral graph is symmetric under the action of the finite permutation group $\mf{S}_4$ and we will show that the modular Laplace operator is closely related to the quadratic Casimir operator of $\mf{sl}(3)$. These two ingredients allow us to use finite group theory and representation theory to deduce certain properties of the spectrum of the Laplace operator acting on the family of tetrahedral modular graph functions. As a by-product we will obtain a simple rederivation of the two-loop results of~\cite{DHoker:2015foa} by the same methods. We note that the modular Laplacian on tetrahedral modular graph functions closes without the need to introduce modular graph \textit{forms} that were recently introduced as a generalisation in~\cite{DHoker:2016mwo,DHoker:2016quv}. Contrary to the cases studied in~\cite{DHoker:2016quv}, the eigenvalues of the modular Laplacian that we obtain are surprisingly not only of the form $s(s-1)$ for non-negative integers $s$. 

There are many possible generalisations and extensions of our work that are beyond the scope of the present paper. A point we have not investigated systematically is to use the inhomogeneous Laplace equations that we find to determine a basis of independent modular graph functions. This point would be very interesting in particular in connection with (elliptic) multiple polylogarithms. It would also be relevant for performing the actual world-sheet integrals over the modular graph functions that we do not attempt here. We note that useful techniques for determining the behaviour of the modular graph functions in the degeneration limit of the world-sheet ($\tau_2\to\infty$) can be found in~\cite{Green:2008uj}. Finally, it would be interesting to consider extensions of the tetrahedral modular graph functions to also include derivatives in such a way that one reconstructs integrands of closed superstring one-loop amplitudes with five and more external legs. A widely open field is also the extension to higher genus string amplitudes, see~\cite{Berkovits:2005df,Berkovits:2005ng,DHoker:2005jc,DHoker:2005ht,Gomez:2010ad,DHoker:2013eea,DHoker:2014gfa,Gomez:2015uha,Pioline:2015qha,Pioline:2015nfa} for some relevant work on genus-two Riemann surfaces, in particular in connection with the so-called Kawazumi--Zhang invariant.

The structure of this article is as follows. We first review in section~\ref{sec:G1} general facts about genus-one amplitudes in closed superstring theory in order to motivate the types of Feynman diagrams and Laplace equations that we analyse. This includes an exposition of the diagrammatical tools for manipulating modular graphs. In section~\ref{sec:TMF}, we introduce the tetrahedral modular graph functions that are the central objects in this paper. We present their Laplace equation in general and introduce a generating function that makes it possible to connect to the representation theory of $\mf{sl}(3)$. We present detailed examples of Laplace equations and spectral properties up to weight $12$ together with some general considerations. These are the main results of this paper. Appendices contain results on simpler two-loop modular graph functions and technical details of some of the calculations of section~\ref{sec:TMF}.

\section{Genus-one amplitudes and modular graph functions}
\label{sec:G1}

We shall consider genus-one contributions to $n$-graviton scattering amplitudes in type II string theory compactified on a torus $T^{d}$ from ten to $D=10-d$ space-time dimensions. The moduli dependence of these contributions appears generically through integrals of the type
\be
\label{eq:genI}
	\cI_{F} = \int_{\cF} d\mu \, F(\bt, \t) \, \G_{(d,d)} \, .
\ee
The integration domain $\cF$ is a fundamental domain of the moduli space, of genus-one Riemann surface
\begin{equation}
\label{eq:FD}
	\cF = \left\{ |\t_1|\leq \frac12, | \t | \geq 1 \right\} = \mathbb{H} / PSL_2(\mathbb{Z})\, ,
\end{equation}
where $\mathbb{H}=\left\{ \tau =\tau_1+i\tau_2 \in \mathbb{C} \, \middle|\, \tau_2>0\right\}$ is the complex upper half plane on which the modular group $PSL(2,\mathbb{Z})$ acts by the standard fractional linear transformation.
The integration measure $d \mu$ is the standard $PSL(2,\Z)$ invariant measure
\begin{equation}
	d \mu = \frac{d \tau_1 d \tau_2}{ \tau_2^2} \, ,
\end{equation}
such that the volume of the fundamental domain is normalised to be $\int_{\cF} d \mu = \frac{\pi}{3}$. The function $\G_{(d,d)}$ is the Narain genus-one partition function~\cite{Narain:1985jj} for the self-dual lattice that describes toroidal compactifications from ten dimensions to $10-d$ dimensions:
\begin{equation}
\label{eq:Narain}
	\G_{(d,d)} = \t_{2}^{d/2} \sum_{n^{I} \in \mathbb{Z}^{2d}} e^{- \pi \tau_2 | Z(n) |^2} e^{\pi i \tau_1 n^I \eta_{IJ} n^J}\, ,
\end{equation}
in terms of the $SO(d,d)$ invariant metric $\eta_{IJ}=\left(\begin{smallmatrix} 0&1\\1&0\end{smallmatrix}\right)$. The mass squared $| Z(n) |^2$ appearing in~\eqref{eq:Narain} is given by
\begin{equation}
	|Z(n)|^{2} = Z_{a}(n) Z^{a}(n) = n^{I} \cV_{a I} \cV^{a}{}_{J} n^{J} \, .
\end{equation}
$\cV^{a}{}_{I}$ is the coset representative parametrising the symmetric space $SO(d,d)/SO(d) \times SO(d)$, which transforms from the left under the local compact subgroup and from the right under the global $SO(d,d)$. Local coordinates can be chosen in terms of the metric and $B$-field on the torus in the standard fashion~\cite{Narain:1985jj,Maharana:1992my}. The Narain partition function $\G_{(d,d)}$ is invariant under $SO(d,d,\Z)$ transformations and modular transformations from $PSL(2,\Z)$. The invariance under $SO(d,d,\Z)$ is obvious, $PSL(2,\Z)$ invariance can only be seen after using Poisson resummation. The integral~\eqref{eq:genI} is by construction still a function of the moduli $\cV^a{}_I$ of the space-time Narain torus; the dependence on the world-sheet torus parameter $\tau$ is being integrated over.

The function $F(\t,\bt)$ appearing in~\eqref{eq:genI} encodes the specifics of the scattering process under consideration. It is required to be invariant under $PSL(2, \Z)$ transformations acting by
\begin{equation}
	F(\t, \bt) = F \biggl(\frac{ a\t + b}{ c\t + d}, \frac{ a\bt + b}{c\bt + d} \biggr) \, ,  \qquad a d - b c = 1 \, ,
\end{equation}
with $a,b,c,d\in\mathbb{Z}$. For general processes $F(\tau,\bar{\tau})$ term will be a complicated function encoding momentum and Narain moduli dependence. However, for a four-graviton interaction, its form can be found explicitly as a Koba--Nielsen prefactor~\cite{Green:1981yb,DHoker:2015foa}
\be
\label{eq:4pt}
F(\tau,\bar{\tau}) = \prod_{i = 1}^{4} \int_{\Sigma} \frac{d^2 z_i}{\tau_2} e^{\mathcal{D}} \, ,
\ee
where $\mathcal{D}$ is a sum over all insertions at local coordinates $z_i$ and $z_j$
 \be
 \label{eq:defD}
\mathcal{D} = \sum_{i < j} s_{i j} G(z_i - z_j | \tau) \, 
 \ee
with dimensionless Mandelstam variables
\be
\label{eq:MS}
s_{i j} = - \frac{1}{4} \alpha^{\prime} (k_i +  k_j)^2 
\ee
and $G(z_i-z_j|\tau)$ the translation invariant scalar propagator between $z_i$ and $z_j$ on the world-sheet torus of modulus $\tau$. We will give an explicit form for the propagator below in~\eqref{eq:Pmom}. The integral in~\eqref{eq:4pt} is over the world-sheet of the torus $\Sigma$ that we parametrise in a fixed domain of volume $\int_\Sigma d^2z = \tau_2$, where $d^2z = dz_1 dz_2$ in terms of the real and imaginary parts of $z=z_1+iz_2$. While the formula~\eqref{eq:4pt} is correct for four-graviton scattering, one will have additional insertions beyond Koba--Nielsen factors for higher point amplitudes~\cite{Green:2013bza,Mafra:2016nwr} and there can be additional Narain moduli dependences. 

The integral $\mathcal{I}_F$ in~\eqref{eq:genI} is an object of central interest in string theory. However, no closed formula for it is known. From a low-energy perspective, one can consider the $\alpha'$-expansion of the integral, corresponding to an expansion in $s_{ij}\ll 1$. This generates analytical (in $\alpha'$) terms in the scattering amplitude~\cite{Green:1999pv,Green:2008uj}. Sometimes one can then understand the integrand, and possibly even the integral, at a fixed order in $\alpha'$ by studying the differential equations the integrand satisfies. The integral $\mathcal{I}_F$ is a function on Narain moduli space $SO(d,d)/SO(d)\times SO(d)$ and for the case when $F(\t,\bt)$ is independent of the Narain moduli one can compute the action of the $SO(d,d)$ Laplacian by using~\cite{Obers:1999um}
\begin{align}
\left[\Laplace_{SO(d,d)} -2 \Laplace_{SL(2)} +\frac14d(d-2)\right] \Gamma_{(d,d)} = 0
\end{align}
that relates the $SO(d,d)$ action to one of the modular invariant $SL(2)$ Laplace operator acting on $\tau$. By partial integration the action of the $SO(d,d)$ Laplacian can then be mapped to the $SL(2)$ Laplacian acting on $F(\t,\bt)$. This action of $\Laplace_{SL(2)}$, more specifically in an $\alpha'$ expansion, is what we shall study in this paper. For genus-one world sheets with metrics parametrised by $\tau\in\mathcal{F}$ of~\eqref{eq:FD} the $SL(2)$ invariant Laplacian on the upper half plane is
\begin{equation}
	\Laplace \equiv \Laplace_{SL(2)} =  4 \tau_2^2 \partial_{\t} \partial_{\bt} = \tau_2^2 \left(\partial_{\tau_1}^2 + \partial_{\tau_2}^2\right)
\end{equation}
and we shall henceforth drop the subscript $SL(2)$ on the Laplacian as it is the only one we will use. 

\subsection{Low-order contributions to the four-graviton amplitude}

{}From the definition of $\mathcal{D}$ in~\eqref{eq:defD} and the Mandelstam invariants we see that we can perform a low-energy expansion in the $s_{ij}$ corresponding to the $\alpha'$ expansion of string theory. The result of expanding the exponential in~\eqref{eq:4pt} can be represented by a world-sheet Feynman diagram consisting of four points that are connected by $n$ lines, where $n$ is the order of the expansion in Mandelstam variables. 

The function $F(\t,\bt)$ controlling the four-graviton amplitude expands for small momenta $s_{i j} \ll 1$ (corresponding to a power expansion in $\alpha'$) as
 \be
 \label{eq:Fexp}
 F(\tau,\bar{\tau}) = \sum_{n = 0}^{\infty} \frac{1}{n!} \prod_{i = 1}^{4} \int_{\Sigma} \frac{d^2 z_i}{\tau_2} \biggl( \sum_{j < k} s_{j k} G(z_j - z_k| \tau) \biggr)^n \, ,
 \ee
where $n$ counts the number of world-sheet propagators between the $4$ points $z_i$. By virtue of the definition~\eqref{eq:MS} of the $s_{ij}$ this corresponds also to the power of $\alpha'$. Because of momentum conservation the $s_{ij}$ are not all independent in the massless four-point amplitude. Letting 
\begin{align}
s= s_{12}=s_{34} \,,\quad t = s_{13}=s_{24} \,,\quad t=s_{23}=s_{14}
\end{align}
one has
\be
s + t + u = 0 \, .
\ee
As a consequence, the analytic part of the four-graviton amplitude can be expanded in a double series in $s^2+t^2+u^2$ and $s^3+t^3+u^3$, except for the classical Einstein--Hilbert contribution~\cite{Green:1999pv}. 

The low order terms in the expansion of $F$ can be found for example in~\cite{Green:2008uj} along with their integrals contributing to the four-graviton amplitude in the low-energy expansion. We recall that the one-loop string theory calculation in ten space-time dimensions gives contributions to $R^4$ and then to every even derivative order starting from $D^6R^4$. (In lower dimensions one also has contributions for $D^4R^4$.) The integrated contributions up to $D^{10}R^4$ have been worked out~\cite{DHoker:2015foa}. A further discussion of the low order contributions can be found in~\cite{Green:1999pv,Green:2008uj}.

\subsection{Laplacian on modular graph functions and the Green's function}

For high order in $\alpha'$ the explicit functional dependence of $F$ on $\tau$ is not very well understood. As in~\cite{DHoker:2015foa}, one can consider the function~\eqref{eq:Fexp} that generates the world-sheet diagrams in the $\alpha'$-expansion as a prototype of a new class of functions called \textit{modular graph functions} that are constructed from world-sheet Feynman diagrams with an arbitrary number of points (not only four) connected by scalar propagators. These diagrams will not directly correspond to string processes but can serve as an interesting class of modular functions and certainly are relevant to the string theory calculation. 

In view of the structure of~\eqref{eq:Fexp} that is given by an integral over products Green's function connecting different vertices we will follow~\cite{DHoker:2015foa,DHoker:2015qmf} and study the following more general class of functions where the integrand is given by
\be
\label{eq:IMG}
\cI(\{n_{kl}\},\tau,\bar{\tau}) = \prod_{i = 1}^{n} \int_{\Sigma} \frac{d^{2} z_{i}}{\tau_2} \biggl( \prod_{k < l} G^{n_{kl}}(z_{k} - z_{l} | \tau) \biggr) \, 
\ee
for some non-negative integer powers $n_{kl}$ and a total of $n$ vertices. Compared to~\eqref{eq:Fexp} we have removed the Mandelstam variables and allowed for an arbitrary number $n$ of vertices. Dimensionally, such an integrand would be related to an amplitude with low-energy action of the form $D^{2w} R^4$, where $w=\sum_{k<l} n_{kl}$ is the weight of the integrand. However, this is generally only true dimensionally as it is known that for $n>4$ the integrand of the genus one amplitude is not of this simple form but also involves derivatives of Green's functions~\cite{Green:2013bza}. 

Functions arising from modular integrals over expressions of the type~\eqref{eq:IMG} are called \textit{modular graph functions}. As we will review below, they can be represented graphically in terms of Feynman world-sheet diagrams. Modular graph functions are invariant under modular transformations acting on $\tau$ and can appear as constituents of higher derivative corrections. Besides this physical relevance, they represent an interesting new class of modular functions on the upper half plane and we will be interested in evaluating the modular Laplacian acting on them. 

We shall use a diagrammatic way of computing the Laplacian acting on a modular graph function, similar to~\cite{Basu:2016kli}. The basic tool is to first rewrite the modular Laplacian using deformation theory as~\cite{DHoker:2015foa}
\begin{align}
\label{eq:Ldef}
\Laplace = 4 \tau_2^2 \partial_\t \partial_{\bt} = \delta_\mu \bar{\delta}_\mu\,,
\end{align}
where $\delta_\mu$ denotes infinitesimal changes in the complex structure while keeping the coordinates fixed. The advantage of this formalism is that one can work out the deformation of the Green's function on general grounds. As shown in~\cite{DHoker:2015foa},
the action of the deformation $\delta_\mu$ acting on a single Green's function connecting two points $z_i$ and $z_j$ can be replaced by the insertion of an additional vertex:
\be
\label{eq:deltaG}
\delta_\mu G( z_{i} - z_{j}|\tau) = - \frac{1}{\pi} \int_{\Sigma} d^{2} z \partial_{z} G(z - z_{i}|\tau) \partial_{z} G({z - z_{j}|\tau}) \, .
\ee
Moreover, the Laplacian with respect to the modular parameter satisfies
\be
\Laplace G( z_{i} - z_{j}|\tau) = 0 \, ,
\ee
since the result becomes a total derivative. Therefore, when evaluating the modular Laplacian on a product of Green's functions one has to apply the deformations $\delta_\mu$ and $\bar{\delta}_\mu$ to different factors.

 Obviously the function $\cI(\tau,\bar{\tau})$ has in general highly complicated dependence on
$\t, \bt$, however some very simple and elegant answers were found in the past. For the purpose of understanding these results we can express $G$ through 
a lattice sum:
\be
\label{eq:Pmom}
G(z| \tau) = \sum_{(m,n) \in \mathbb{Z}^2\setminus\{(0,0)\}} \frac{\tau_{2}}{\pi |m\tau  + n|^{2}} e^{\frac{\pi}{\tau_2} (\bar{z}(m \tau + n) - z(m \bar{\tau} + n))} \, ,
\ee
where the integers $m, n$ parametrise the discrete momenta on a torus $p = m + n \tau$. The zero momentum was removed from the lattice sum $\mathbb{Z}^{2}$. In this representations the modular invariance of $G(z| \tau)$ can be easily seen. 
The scalar Green's functions on a world-sheet torus $\Sigma$ of modulus $\t$ satisfies the identities
\be \label{eq:G_id}
\partial_{z} \partial_{\bar{z}} G(z | \tau) = - \pi \delta^{(2)}(z) + \frac{\pi}{\tau_2}  \, , \qquad \int_{\Sigma} d^{2} z G(z| \tau) = 0 \, ,
\ee
where the second condition is related to the choice of zero mode. The zero mode does not contribute to~\eqref{eq:Fexp} by momentum conservation $\sum_{i<j} s_{ij}=0$.

It is very helpful to represent the Green's functions and their derivatives in a graphical way in order to simplify the calculations. These are the modular graphs that represent the integrands of the modular integrals. 

A single point is a symbol for the integration of an insertion and lines represent the Green's functions between two insertions
\be
\begin{tikzpicture}
\filldraw[black] (-7,-0.5) circle (2pt) node[anchor=north] {$i$} node[anchor=east]{$\frac{1}{\tau_2} \int_{\Sigma} d^2 z_{i} = \;$};  
\draw[gray, thick] (-1,-0.5) -- (1,-0.5);
\filldraw[black] (-1,-0.5) circle (2pt) node[anchor=north] {$i$} node[anchor=east]{$\hspace{-4 mm}\, , \, \hspace{4 mm}\frac{1}{\tau_2^2} \int_{\Sigma} d^2 z_{i} d^2 z_{j} \, G(z_{i} - z_{j}|\tau) \; = \; $}; 
\filldraw[black] (1,-0.5) circle (2pt) node[anchor=north] {$j$}; 
\node at (1.3,-0.5) {$.$}; 
\end{tikzpicture} 
\ee
Due to reflection invariance in $z$ one does not need to put arrows on the propagators. 

In the same way several lines joining at one point $z$ mean that several Green's functions connect  an insertion point $z$ to various other insertions and all of them are
integrated out at the end
\be
\begin{tikzpicture}

\draw[gray, thick] (-1,-0.5) -- (1,0.5);
\draw[gray, thick] (-1,-0.5) -- (1,0);
\draw[gray, thick] (-1,-0.5) -- (1,-0.5);
\draw[gray, thick] (-1,-0.5) -- (1,-1);
\draw[gray, thick] (-1,-0.5) -- (1,-1.5);
\filldraw[black] (-1,-0.5) circle (2pt) node[anchor=north] {$z$} node[anchor=east]{$\frac{1}{\tau_2^{n + 1}} \int_{\Sigma} d^2 z \prod_{i = 1}^{n} \int_\Sigma d^2 z_{i} \, G(z - z_{i}|\tau) \; = \; $}; 

\filldraw[black] (1,0.5) circle (2pt) node[anchor=north] {$\scriptstyle n$};
\filldraw[black] (1,0) circle (2pt) node[anchor=north] {$\scriptstyle n-1$}; 
\filldraw[black] (1,-0.5) circle (2pt) node[anchor=north] {$\scriptstyle ... \qquad .$}; 
\filldraw[black] (1,-1) circle (2pt) node[anchor=north] {$\scriptstyle 2$}; 
\filldraw[black] (1,-1.5) circle (2pt) node[anchor=north] {$\scriptstyle 1$}; 
\end{tikzpicture}
\ee
We also introduce the action of the derivative in $z$ acting on one of the Green's functions
\be
\begin{tikzpicture}
\draw[gray, thick] (-1,-0.5) -- (1,-0.5);
\filldraw[black] (-1,-0.5) circle (2pt) node[anchor=north] {$i$} node[anchor=south west] {$\,\,\partial$}  node[anchor=east]{$\frac{1}{\tau_2^2} \int_{\Sigma} d^2 z_{i} d^2 z_{j} \, \partial_{z_{i}} G(z_{i} - z_{j}|\tau) \; = \; $}; 
\filldraw[black] (1,-0.5) circle (2pt) node[anchor=north] {$j$} ;
\node at (-0.2,-1.5) {$=-$};
\draw[gray, thick] (0.3,-1.5) -- (2.3,-1.5);
\filldraw[black] (0.3,-1.5) circle (2pt) node[anchor=north] {$i$};
\filldraw[black] (2.3,-1.5) circle (2pt) node[anchor=north] {$j$} node[anchor=south east] {$\,\,\partial$} node[anchor=west]{$ \; = \; - \frac{1}{\tau_2^2} \int_{\Sigma} d^2 z_{i} d^2 z_{j} \, \partial_{z_{j}} G(z_{i} - z_{j}|\tau) \, .$} ;
\end{tikzpicture}
\ee
Here, we have also illustrated the consequence of translation invariance $\partial_{z_i} G(z_i-z_j|\tau) = - \partial_{z_j} G(z_i-z_j|\tau)$.

Last but not least we are always able to rewrite the action of the deformation $\delta_\mu$ into the action of derivatives on the world-sheet by introducing an additional insertion, such that 
\be
\label{eq:deltaGDiag}
\begin{tikzpicture}
\draw[gray, thick] (-1,-0.5) -- (1,-0.5);
\filldraw[black] (-1,-0.5) circle (2pt) node[anchor=north] {$i$} node[anchor=south west] {$\qquad \delta_\mu$}  node[anchor=east]{$\delta_\mu \biggl[ \frac{1}{\tau_2^2} \int_{\Sigma} d^2 z_{i} d^2 z_{j} G(z_{i} - z_{j}|\tau) \biggr] \; = \; $}; 
\filldraw[black] (1,-0.5) circle (2pt) node[anchor=north] {$j$} node[anchor=west] {$\; = \;$}; 
\draw[gray, thick] (-2.5,-1.5) -- (-0.5,-1.5);
\filldraw[black] (-2.5,-1.5) circle (2pt) node[anchor=north] {$i$} node[anchor=east] {$\quad - \frac{\tau_2}{\pi} $};
\filldraw[black] (-1.5,-1.5) circle (2pt) node[anchor=north] {$z$} node[anchor= south east] {$\partial \,$} node[anchor= south west] {$\, \partial$};
\filldraw[black] (-0.5,-1.5) circle (2pt) node[anchor=north] {$j$} node[anchor=west]{$\; = \; - \frac{1}{\pi \tau_2^2} \int_{\Sigma} d^{2} z_{i} d^{2} z_{j} d^{2} z \partial_{z} G(z - z_{i}|\tau) \partial_{z} G({z - z_{j}|\tau}) \, . $} ;
\end{tikzpicture}
\ee
This rule is due to~\eqref{eq:deltaG}. There is a similar formula for the conjugate deformation $\bar{\delta}_\mu$ in terms of the conjugate world-sheet derivative $\bar{\partial}_z$. 

Because of equation \eqref{eq:G_id} we see that every diagram with at least one node that has only one Green's function connecting to it is vanishing. Therefore for tadpoles diagrams we obtain
\be
\label{eq:tadpole}
\begin{tikzpicture}
\draw[gray, thick, dashed] (-2,-0.5) -- (0,0);
\draw[gray, thick, dashed] (-2,0) -- (0,0);
\draw[gray, thick] (0,0) -- (2,0);
\draw[gray, thick, dashed] (-2,0.5) -- (0,0);
\filldraw[black] (0,0) circle (2pt) node[anchor=north] {$i$};
\filldraw[black] (2,0) circle (2pt) node[anchor=north] {$z$} node[anchor=west] {$\; \propto \; \frac{1}{\tau_2} \int_{\Sigma} d^{2} z G(z_{i} - z| \tau) = 0 \, .$};
\end{tikzpicture}
\ee
Additionally, we read out from \eqref{eq:G_id} the diagrammatical simplification rule
\be
\label{eq:ddcontract}
\bp
\draw[gray, thick, dashed] (-3,-0.5) -- (-2,0);
\draw[gray, thick, dashed] (-3,0) -- (-2,0);
\draw[gray, thick, dashed] (-3,0.5) -- (-2,0);
\draw[gray, thick, dashed] (0,0) -- (1,-0.5);
\draw[gray, thick, dashed] (0,0) -- (1,0);
\draw[gray, thick, dashed] (0,0) -- (1,0.5);
\draw[gray, thick] (-2,0) -- (0,0);
\filldraw[black] (-2,0) circle (2pt) node[anchor=north] {$i$} node[anchor=south west] {$\partial$} ;
\filldraw[black] (0,0) circle (2pt) node[anchor=north] {$j$} node[anchor=south east] {$\bar{\partial}$};
\filldraw[black] (1.5,0) circle (0pt) node[anchor=center] {$\; = \frac{\pi}{\tau_2} $};
\draw[black,thick] (2.1,-0.6) -- (2.0,-0.6) -- (2.0,0.6) -- (2.1,0.6);
\draw[black,thick] (4.1,-0.6) -- (4.2,-0.6) -- (4.2,0.6) -- (4.1,0.6);
\draw[gray, thick, dashed] (2.1,-0.5) -- (3.1,0);
\draw[gray, thick, dashed] (2.1,0) -- (3.1,0);
\draw[gray, thick, dashed] (2.1,0.5) -- (3.1,0);
\draw[gray, thick, dashed] (3.1,0) -- (4.1,-0.5);
\draw[gray, thick, dashed] (3.1,0) -- (4.1,0);
\draw[gray, thick, dashed] (3.1,0) -- (4.1,0.5);
\filldraw[black] (3.1,0) circle (2pt) node[anchor=north] {$i = j$};
\draw[gray, thick, dashed] (5.2,-0.5) -- (6.2,0);
\draw[gray, thick, dashed] (5.2,0) -- (6.2,0);
\draw[gray, thick, dashed] (5.2,0.5) -- (6.2,0);
\draw[gray, thick, dashed] (7.2,0) -- (8.2,-0.5);
\draw[gray, thick, dashed] (7.2,0) -- (8.2,0);
\draw[gray, thick, dashed] (7.2,0) -- (8.2,0.5);
\draw[black,thick] (5.2,-0.6) -- (5.1,-0.6) -- (5.1,0.6) -- (5.2,0.6);
\draw[black,thick] (8.2,-0.6) -- (8.3,-0.6) -- (8.3,0.6) -- (8.2,0.6);
\filldraw[black] (4.7,0) circle (0pt) node[anchor=center] {$\; - \frac{\pi}{\tau_2} $};
\filldraw[black] (6.2,0) circle (2pt) node[anchor=north] {$i$};
\filldraw[black] (7.2,0) circle (2pt) node[anchor=north] {$j$};
\node at (8.4,0) {$.$};
\ep
\ee
The derivative with respect to world-sheet variables acting on one of the Green's functions can be moved on the graph reproducing the integration by parts formula. 
We obtain for example
\be
\bp
\filldraw[black] (0,0) circle (2pt) node[anchor=north] {$i$} node[anchor=south east] {$\partial$} ;
\draw[gray, thick] (-1,0) -- (0,0);
\draw[gray, thick] (0,0) -- (1,0.5);
\draw[gray, thick] (0,0) -- (1,0);
\draw[gray, thick] (0,0) -- (1,-0.5);
\node at (1.5,0) {$= \, \, - $};
\filldraw[black] (3,0) circle (2pt) node[anchor=north] {$i$};
\draw[gray, thick] (2,0) -- (3,0);
\draw[gray, thick] (3,0) -- (4,0.5);
\draw[gray, thick] (3,0) -- (4,0);
\draw[gray, thick] (3,0) -- (4,-0.5);
\node at (3.2,0.35) {$\partial$} ;
\node at (4.5,0) {$-$} ;
\filldraw[black] (6,0) circle (2pt) node[anchor=north] {$i$};
\draw[gray, thick] (5,0) -- (6,0);
\draw[gray, thick] (6,0) -- (7,0.5);
\draw[gray, thick] (6,0) -- (7,0);
\draw[gray, thick] (6,0) -- (7,-0.5);
\node at (6.5,0.2) {$\partial$} ;

\node at (7.5,0) {$-$} ;
\filldraw[black] (9,0) circle (2pt) node[anchor=north] {$i$};
\draw[gray, thick] (8,0) -- (9,0);
\draw[gray, thick] (9,0) -- (10,0.5);
\draw[gray, thick] (9,0) -- (10,0);
\draw[gray, thick] (9,0) -- (10,-0.5);
\node at (9.25,-0.4) {$\partial$} ;
\node at (10.5,0) {$\, .$} ;
\ep
\ee

\subsection{Modular graph functions with one and two world-sheet loops}

Some subsets of modular functions with a particular geometric structures are well understood. For example a simple $s$-polygon of Green's functions 
reproduces non-holomorphic Eisenstein series with a somewhat unusual, yet for our purpose useful, normalisation~\cite{DHoker:2015foa}:
\be
\bp
\coordinate (A1) at (0,0);
\coordinate (A2) at (0.75,0.5);
\coordinate (A3) at (1.5,0);
\coordinate (A4) at (1.5,-0.75);
\coordinate (A5) at (0.75,-1.25);
\coordinate (An) at (0,-0.75);
\draw[gray, thick] (An) -- (A1) -- (A2) -- (A3) -- (A4) -- (A5);
\draw [dashed] (A5) -- (An);

\node [above left] at (0.45,0.15) {$\scriptstyle 1$};
\node [above right] at (1, 0.2) {$\scriptstyle 2$};
\node [right] at (1.45,-0.35) {$\scriptstyle 3$};
\node [right] at (1,-1.1) {$\scriptstyle 4$};
\node [left] at (0,-0.4) {$\scriptstyle s$};

\filldraw[black] (A1) circle (2pt);
\filldraw[black] (A2) circle (2pt);
\filldraw[black] (A3) circle (2pt);
\filldraw[black] (A4) circle (2pt);
\filldraw[black] (A5) circle (2pt);
\filldraw[black] (An) circle (2pt);

\node [right] at (1.7, -0.75/2) {$\displaystyle \, = \, \sum_{(m, n) \in \mathbb{Z}^2/\{0\}} \frac{\tau_2^{s}}{\pi^{2 s} |m + n \tau|^{2 s}} = E_{s}(\tau, \bar{\tau}) \, .$};
\ep
\ee
This is the simplest non-trivial structure that appears as the modular graph function, with a single summation over the discretised momentum in the loop 
and it depends just on a single value $s$, that is the number of internal vertex insertions. Eisenstein functions are know to satisfy a homogeneous Laplace equation
\be
( \Laplace - s(s - 1) ) E_{s}(\tau,\bar{\tau}) = 0 \, .
\ee
This equations can also be proved diagrammatically using the rules outlined above. For the particular value of $s = 0$, we use the normalisation
\be
E_{0}(\tau,\bar{\tau}) = 1 \, .
\ee
The next more complex structure was discussed in detail in~\cite{DHoker:2015foa} and depends on a triplet $(s,t,p)$ of integer values, that describe the number of vertex insertions on the path connecting points $z_1$ and $z_2$ on the torus 
\be
\label{eq:Cstp}
\bp
\coordinate (A1) at (0,-0.35);
\coordinate (A11) at (0.75,-0.35);
\coordinate (A12) at (1.75,-0.35);
\coordinate (A13) at (2.75,-0.35);
\coordinate (A14) at (3.75,-0.35);
\coordinate (A2) at (0.75,0.5);
\coordinate (A21) at (1.75,0.5);
\coordinate (A22) at (2.75,0.5);
\coordinate (A23) at (3.75,0.5);
\coordinate (A24) at (4.5,-0.35);
\coordinate (A3) at (4,0.5);
\coordinate (An) at (0.75,-1.25);
\coordinate (An1) at (1.75,-1.25);
\coordinate (An2) at (2.75,-1.25);
\coordinate (An3) at (3.75,-1.25);
\draw[gray, thick] (A1) -- (A11) -- (A12);
\draw[gray, thick] (A1) -- (A2) -- (A21);
\draw[gray, thick] (A22) -- (A23) -- (A24);
\draw[gray, thick] (A1) -- (An) -- (An1) ;
\draw[gray, thick] (An2) -- (An3) -- (A24) ;

\draw[gray, thick] (A13) -- (A14) -- (A24);
\draw[gray, thick, dashed] (A21) -- (A22);
\draw[gray, thick, dashed] (An1) -- (An2);
\draw [dashed] (A12) -- (A13);
\node [left] at (A1) {$\scriptstyle z_1$};

\node [above] at (0.3,0.1) {$\scriptstyle 1$};
\node [above] at (1.25,0.5) {$\scriptstyle 2$};
\node [above] at (3.25,0.5) {$\scriptstyle s-1$};
\node [above] at (4.25,0.1) {$\scriptstyle s$};

\node [above] at (0.45,-0.4) {$\scriptstyle 1'$};
\node [above] at (1.25,-0.4)  {$\scriptstyle 2'$};
\node [above] at  (3.3,-0.4) {$\scriptstyle (t-1)'$};
\node [above] at  (4.1,-0.35) {$\scriptstyle t'$};

\node [below] at (0.2,-0.7) {$\scriptstyle 1''$};
\node [below] at (1.25,-1.2) {$\scriptstyle 2''$};
\node [below] at (3.25,-1.2){$\scriptstyle (p-1)''$};
\node [below right] at  (4.1,-0.6) {$\scriptstyle p''$};

\node [right] at (A24) {$\scriptstyle z_2$};

\node [right] at (A24) {$\quad = \sum^{'}_{\substack{(m_1,n_1)\\(m_2,n_2)}} \frac{\tau_2^{s + t + p}}{\pi^{s + t + p} |m_1 + n_1 \tau|^{2 s} |m_2 + n_2 \tau|^{2 t} | m_1 + m_2 + (n_1 + n_2)\tau|^{2 p} }$};
\node at (6.3, - 1.5) {$= C_{(s,t,p)}(\tau,\bar{\tau}) \, , $};
\filldraw[black] (A1) circle (2pt);
\filldraw[black] (A2) circle (2pt);
\filldraw[black] (A21) circle (2pt);
\filldraw[black] (A22) circle (2pt);
\filldraw[black] (A23) circle (2pt);
\filldraw[black] (A24) circle (2pt);
\filldraw[black] (A11) circle (2pt);
\filldraw[black] (A12) circle (2pt);
\filldraw[black] (A13) circle (2pt);
\filldraw[black] (A14) circle (2pt);
\filldraw[black] (A5) circle (2pt);
\filldraw[black] (An) circle (2pt);
\filldraw[black] (An1) circle (2pt);
\filldraw[black] (An2) circle (2pt);
\filldraw[black] (An3) circle (2pt);
\ep
\ee
where we sum over discrete momenta $p_1$ and $p_2$ in the loop, excluding zero and have solved overall momentum conservation. The prime on the sum indicates that we have to exclude all zero momentum propagators, \textit{i.e.}, $(m_1,n_1)\neq (0,0)$, $(m_2,n_2)\neq (0,0)$ and $(m_1+m_2,n_1+n_2)\neq (0,0)$. In string theory only non-negative integer values for $s$, $t$ and $p$ arise but as argued in~\cite{DHoker:2015foa} the function $C_{(s,t,p)}$ can be analytically continued to arbitrary complex values of the parameters. We will often suppress the arguments $\tau$ and $\bar\tau$.

Starting now we will use following abbreviations to indicate the number of Green's functions that connect two points in a simple manner
\be
\bp
\coordinate (A1) at (0,0);
\coordinate (A2) at (1,0);
\coordinate (A3) at (2,0);
\coordinate (A4) at (3,0);
\coordinate (A5) at (4,0);

\coordinate (A6) at (5,0);
\coordinate (A8) at (7,0);

\filldraw[black] (A1) circle (2pt);
\filldraw[black] (A2) circle (2pt);
\filldraw[black] (A3) circle (2pt);
\filldraw[black] (A4) circle (2pt);
\filldraw[black] (A5) circle (2pt);

\node [below] at (0.5,0) {$\scriptstyle 1$};
\node [below] at (1.5,0) {$\scriptstyle 2$};
\node [below] at (3.5,0) {$\scriptstyle s $};
\node [right] at (A5) {$\, \, = $};

\draw[gray, thick] (A1) -- (A3);
\draw[gray, thick] (A4) -- (A5);
\draw[gray, thick, dashed] (A3) -- (A4);
\draw[gray, thick,] (A6) -- (A8);

\filldraw[black] (A6) circle (2pt);
\filldraw[black] (A8) circle (2pt);

\node at (7.5,0) {$.$};
\greens{6}{0}{s}
\ep
\ee
With this notation Eisenstein functions and $C_{(s,t,p)}$ can be written in a more graphical way
\be
\bp
\node [left] at (0,0) {$\scriptstyle z_1$};
\filldraw [black] (0,0) circle (2pt);
\draw [black, thick] (0.6,0) circle (0.6 cm);
\greens{1.2}{0}{s} 
\node [right] at (1.2,0) {$\,\,\,\, = E_{s}(\tau,\bar{\tau}) \, ,$};

\begin{scope}[shift={(1,0)}]
\draw [thick] (4,0) .. controls (5,1.35) .. (6,0);
\draw [thick] (4,0) .. controls (5,0) .. (6,0);
\draw [thick] (4,0) .. controls (5,-1.35) .. (6,0);
\node [left] at (4,0) {$\scriptstyle z_1$};
\node [right] at (6,0) {$\scriptstyle z_2$};
\greens{5}{1}{s}
\greens{5}{0}{t}
\greens{5}{-1}{p} 
\point{4}{0}
\point{6}{0}
\node [right] at (6.2,0) {$\,\,\,\, = C_{\three{\hspace{-2mm}s}{\hspace{-2mm}t}{\hspace{-2mm}p}}(\tau,\bar{\tau}) \, .$};
\end{scope}
\ep
\ee
It is obvious that $C_{(s,t,p)}$ is completely symmetric under permutations of the $(s,t,p)$ triplet. Furthermore for specific values of $(s,t,p)$ the function
$C_{(s,t,p)}$ simplifies to a quadratic polynomial in Eisenstein series~\cite{DHoker:2015foa}
\be
\label{eq:CabcSpec}
C_{\three{\hspace{-0mm}s}{\hspace{-0mm}t}{\hspace{-0mm}0}} = E_{s} E_{t} - E_{s + t} \, , \qquad 
C_{\three{\hspace{-0mm}s}{\hspace{-0mm}t}{\hspace{-0mm}$\bf -$1}}  = E_{s - 1} E_{t} + E_{s} E_{t - 1} \, .
\ee
This simplification can be easily seen from the lattice sum representation. 
Unfortunately, the differential equation satisfied by a general $C_{(s,t,p)}$ function is not any more homogeneous and can be derived to be~\cite{DHoker:2015foa}
\begin{align}
\label{eq:LCabc}
( \Laplace - \w) C_{\three{\hspace{-0mm}s}{\hspace{-0mm}t}{\hspace{-0mm}p}} 
&= st \left( C_{\three{\hspace{-0mm}s +1}{\hspace{-0mm}t - 1}{\hspace{-0mm}p}} 
   +  C_{\three{\hspace{-0mm}s -1}{\hspace{-0mm}t + 1}{\hspace{-0mm}p}}
   +  C_{\three{\hspace{-0mm}s +1}{\hspace{-0mm}t + 1}{\hspace{-0mm}p-2}}
   -2  C_{\three{\hspace{-0mm}s +1}{\hspace{-0mm}t}{\hspace{-0mm}p-1}}
   -2  C_{\three{\hspace{-0mm}s}{\hspace{-0mm}t + 1}{\hspace{-0mm}p-1}}\right)\nn\\
& \quad +\textrm{   the two other pairs of lines }\,,
\end{align}
where the eigenvalue $\w$ is given by
\be
\w = s(s - 1) + t(t-1) + p(p-1) \, .
\ee
The spectrum of the modular Laplacian was studied in great detail in~\cite{DHoker:2015foa}. In appendix~\ref{app:2loop}, we present a simple rederivation of the results of that paper based on an application of Molien's theorem combined with some representation theory of $\mf{sl}(2)$. 

\section{Tetrahedral family of modular graph functions}
\label{sec:TMF}

In this section, we introduce the family of modular graph functions associated with the tetrahedral graph and an arbitrary number of propagators on all edges. We determine the inhomogeneous Laplace equation satisfied by such functions and study some degeneration limits. The spectrum of the Laplace operator on tetrahedral modular graph functions is partially determined using generating function techniques.

\subsection{Inhomogeneous Laplace equation for tetrahedral modular graphs}

The next very symmetrical topology with three-valent vertices after the one above is that of a tetrahedron (or Mercedes graph):
\be
\bp
\point{2}{-1}
\point{-2}{-1}
\point{0}{2}
\point{0}{0}
\draw[thick]  (2,-1) -- (-2,-1) -- (0,2) -- (2,-1) -- (0,0) -- (0,2) ;
\draw[thick]   (-2,-1) -- (0,0);
\node [above] at (0,2) {$\scriptstyle z_{1}$};
\node [above right] at (0,0) {$\scriptstyle z_{2}$};
\node [below] at (-2,-1) {$\scriptstyle z_{3}$};
\node [below] at (2,-1) {$\scriptstyle z_{4}$};
\greens{-1}{0.5}{s}
\greens{0}{1}{t}
\greens{1}{0.5}{p}
\greens{1}{-0.5}{q}
\greens{-1}{-0.5}{w}
\greens{0}{-1}{r}
\node [right] at (1,0.5) {$ \quad = \sum^{'}_{p_i} \frac{\tau_2^{s + t+ p + q + w +r}}{\pi^{s + t + p + q + w + r} |p_1 |^{2 s} |p_2 |^{2 t} |p_1 + p_2|^{2 p} |p_3|^{2 q} |p_1 + p_2 + p_3|^{2 r} |p_2 + p_3|^{2 w}} = C_{\six{\hspace{-3 mm}s}{t}{p}{q}{\hspace{-1mm }w}{r}} $};
\ep
\ee
The restriction on the sum means that the discrete momenta $p_i=m_i+n_i \tau$ for integers $m_i,n_i\in \mathbb{Z}$ are all non-zero and similarly for all other propagators, \textit{i.e.}, $p_1+p_2\neq 0$, $p_2+p_3\neq 0$ and $p_1+p_2+p_3\neq 0$. We have already solved momentum conservation in the above expression and the loop momenta are labelled as
\be
\raisebox{-0.3\height}{\begin{tikzpicture}[thick,decoration={
    markings,
    mark=at position 0.5 with {\arrow{>}}}
    ] 
  \draw[postaction=decorate] (0,0)--(0,2);
  \draw[postaction=decorate] (0,0)--(1.73,-1);  
  \draw[postaction=decorate] (-1.73,-1)--(0,0);
  \draw[postaction=decorate] (-1.73,-1)--(0,2);
  \draw[postaction=decorate] (0,0)--(0,2);
  \draw[postaction=decorate] (1.73,-1)--(-1.73,-1);
  \draw[postaction=decorate] (0,2)--(1.73,-1);  
  \draw (-1.2,0.8) node {$p_1$};
  \draw (-0.3,0.7) node {$p_2$};  
  \draw (1.6,0.7) node {$p_1+p_2$};  
  \draw (0.7,-0.1) node {$p_3$};  
  \draw (0,-1.3) node {$p_1+p_2+p_3$};    
  \draw (-0.2,-0.6) node {$p_2+p_3$};    
\end{tikzpicture}}  
\ee
As is well-known, the tetrahedron has point symmetry group $\mf{S}_4$ acting on it. An explicit form of the action of this symmetric group on the graph can be found for example in~\cite{Basu:2014hsa}. For a tetrahedral modular graph function $C\!\!\raisebox{-0.3\height}{\six{s}{t}{p}{q}{w}{r}}$ we will call $s+t+p+q+r+w$ the \textit{weight} of the function. The genuine first non-trivial case arises at weight $6$ and was treated already in~\cite{Basu:2015ayg}. We will re-derive it within our more general analysis below.

Without solving momentum conservation the tetrahedral modular graph function $C\!\!\raisebox{-0.3\height}{\six{s}{t}{p}{q}{w}{r}} $ can be expressed in the symmetric way through six lattice sums and four Kronecker deltas preserving momentum conservation at each vertex $z_{i}$:
\be
C_{\six{\hspace{-2 mm}1}{2}{3}{4}{6}{5}}  = \sump_{(m_{i},n_{i}) \in \mathbb{Z}^{2}} \delta^{z_i}_{m 0} \delta^{z_{i}}_{n 0} \prod_{j = 1}^{6}  \frac{\tau^{s_j}_{2}}{\pi^{s_j} |m_{j} + n_{j} \tau|^{2  s_j}}\,,
\ee
where the labels $1,\ldots,6$ on the left-hand side stand for the parameters $s_1,\ldots,s_6$ appearing on the right-hand side.

Using either graphical methods or the sum representation, we can evaluate the modular Laplacian on these tetrahedral modular graph functions to be
\begin{align}
\label{eq:LC6}
\bigl(\Laplace - \omega \bigr)  C_{\six{\hspace{-2 mm}s}{t}{p}{q}{w}{r}}  &=st \Bigg(
    C_{\six{\hspace{-2 mm}$\bf  -$1}{1}{}{}{}{}} 
  +C_{\six{1}{\hspace{-2 mm}$\bf  -$1}{}{}{}{}} 
  +C_{\six{1}{1}{\hspace{-2 mm}$\bf  -$2}{}{}{}} 
  -2C_{\six{1}{}{\hspace{-2 mm}$\bf  -$1}{}{}{}} 
  -2C_{\six{}{1}{\hspace{-2 mm}$\bf  -$1}{}{}{}}   
\Bigg)\nn\\
&\hspace{10mm}\textrm{+ 11 other adjacent pairs of lines}\nn\\
&\quad+sq\Bigg(
    C_{\six{1}{\hspace{-2 mm}$\bf  -$2}{}{1}{}{}} 
  +C_{\six{1}{}{\hspace{-2 mm}$\bf  -$2}{1}{}{}} 
  +C_{\six{1}{}{}{1}{\hspace{-2 mm}$\bf  -$2}{}} 
  +C_{\six{1}{}{}{1}{}{\hspace{-2 mm}$\bf  -$2}}   
  +2 C_{\six{1}{}{\hspace{-2 mm}$\bf  -$1}{1}{\hspace{-2 mm}$\bf  -$1}{}} 
  +2 C_{\six{1}{\hspace{-2 mm}$\bf  -$1}{}{1}{}{\hspace{-2 mm}$\bf  -$1}}    \nn\\
  &\hspace{15mm}
    -2C_{\six{1}{}{\hspace{-2 mm}$\bf  -$1}{1}{}{\hspace{-2 mm}$\bf  -$1}} 
    -2C_{\six{1}{\hspace{-2 mm}$\bf  -$1}{\hspace{-2 mm}$\bf  -$1}{1}{}{}} 
    -2C_{\six{1}{}{\hspace{-2 mm}$\bf  -$1}{1}{}{\hspace{-2 mm}$\bf  -$1}} 
    -2C_{\six{1}{}{}{1}{\hspace{-2 mm}$\bf  -$1}{\hspace{-2 mm}$\bf  -$1}} 
\Bigg) \nn\\
&\hspace{10mm}\textrm{+ 2 other opposite pairs of lines}
\end{align}
with the `eigenvalue' $\w$ being
\be
\label{eq:om}
\w = s (s - 1) + p (p-1) + q (q-1) + r (r-1) + t (t-1) + w (w-1)-2(w p + q s + r t)
\ee
The final mixed term does not arise for the `sunset' functions of~\cite{DHoker:2015foa} and is formed as the sum over the three sets of non-adjacent (opposite) lines in the diagram. The notation in~\eqref{eq:LC6} means that the indices on the corresponding lines of the diagrams are increased or decreased in the indicated places while maintaining the labels on the left-hand side of the equation. Thus
\begin{align}
\label{eq:6short}
\raisebox{-0.5\height}{\scalebox{1.5}{\six{\hspace{-2 mm}$\bf  -$1}{}{}{}{1}{}}}  \hspace{2mm}\equiv \raisebox{-0.5\height}{\scalebox{1.5}{\six{\hspace{-3 mm}s$\bf -$1}{t}{p}{q}{w\!$\bf +$\!1}{r}}}
\end{align}
We present some details on the derivation of~\eqref{eq:LC6} in appendix~\ref{app:LC6}.

\subsection{Degeneration limits}

Below we will also require some degeneration limits of the tetrahedral graphs when some of the vertices come together. These are
\be
\bp
\point{1}{-1}
\point{-1}{-1}
\point{0}{1}

\draw[thick]  (-1,-1) -- (1,-1);

\draw [thick] (-1,-1) .. controls (-1.4,0) .. (0,1);
\draw [thick] (-1,-1) .. controls (-0.6,0) .. (0,1);

\draw [thick] (1,-1) .. controls (1.4,0) .. (0,1);
\draw [thick] (1,-1) .. controls (0.6,0) .. (0,1);

\node [above] at (0,1) {$\scriptstyle z_{1}$};

\node [below] at (-1,-1) {$\scriptstyle z_{2}$};

\node [below] at (1,-1) {$\scriptstyle z_{3}$};

\greens{-1.2}{0}{s}
\greens{-0.5}{0}{t}

\greens{0.5}{0}{p}
\greens{1.2}{0}{q}

\greens{0}{-1}{r}

\node [right] at (1,0) {$\quad \,\,\, = \sum^{'}_{p_i} \frac{\tau_2^{s + t + p + q + r}}{\pi^{s + t + p + q + r} |p_1 |^{2 s} |p_2 |^{2 t} |p_3|^{2 p} |p_1 + p_2|^{2 r} |p_1 + p_2 + p_3|^{2 q} } =  C_{\five{\hspace{-2mm}s}{\hspace{-2mm}t}{\hspace{-2mm}p}{\hspace{-2mm}q}{r}}(\tau,\bar{\tau})\, , $};
\ep  
\ee

\be
\bp
\point{0}{-1}
\point{0}{1}

\draw [thick] (0,-1) .. controls (-1.4,0) .. (0,1);
\draw [thick] (0,-1) .. controls (-0.6,0) .. (0,1);

\draw [thick] (0,-1) .. controls (1.4,0) .. (0,1);
\draw [thick] (0,-1) .. controls (0.6,0) .. (0,1);

\node [above] at (0,1) {$\scriptstyle z_{1}$};

\node [below] at (0,-1) {$\scriptstyle z_{2}$};

\greens{-1.1}{0}{s}
\greens{-0.5}{0}{t}

\greens{0.5}{0}{p}
\greens{1.1}{0}{q}

\node [right] at (1,0) {$\quad \,\,\, = \sum^{'}_{p_i} \frac{\tau_2^{s + t + p +q}}{\pi^{s + t + p +q} |p_1 |^{2 s} |p_2 |^{2 t} |p_3|^{2 p} |p_1 + p_2 + p_3|^{2 q} } = C_{\four{\hspace{- 3 mm} s}{\hspace{- 3 mm} t}{\hspace{- 3 mm} p}{\hspace{- 3 mm} q}}(\tau,\bar{\tau})$};
\ep 
\ee 

\allowdisplaybreaks{
As in~\eqref{eq:CabcSpec}, setting one of the values in $C\!\!\raisebox{-0.3\height}{\six{s}{t}{p}{q}{w}{r}} $ to the value $0$ or $-1$ leads to a simplification in the modular graph functions. Some of simplifications are multi-term identities.
\begin{subequations}
\label{eq:simpT}
\begin{align}
C_{\six{\hspace{-2 mm}s}{t}{p}{q}{w}{0}} &=  C_{\five{s}{p}{w}{q}{t}} -   C_{\three{s + w}{q + p}{t}}\,,\\
 \label{eq:simp1}
C_{\six{\hspace{-2 mm}s}{t}{p}{q}{w}{$\bf{-}$1}}  + C_{\six{\hspace{-2 mm}s}{t}{p}{w}{q}{$\bf{-}$1}} &=  C_{\five{\,\,\,s$\bf-$1}{\,\,\,\,\,p}{w}{q}{t}}  + C_{\five{s}{\hspace{-2mm}p}{\hspace{-2mm}w$\bf-$1}{\,q}{t}} + C_{\five{s}{\hspace{2mm}p$\bf-$1}{w}{q}{t}} + C_{\five{s}{\hspace{-2mm} p}{\hspace{-9mm}w}{\hspace{-4mm}q$\bf-$1}{t}}
 - C_{\five{s}{p}{w}{q}{t$\bf-$1}}\,,\\
C_{\five{s}{p}{q}{r}{0}} &=  C_{\four{s}{p}{q}{r}} - E_{s + p} E_{q + r} \,,\\
C_{\five{0}{s}{p}{q}{r}} &= E_{s} C_{\three{r}{p}{q}} - C_{\three{s + r}{p}{q}} \,,\\
C_{\five{s}{p}{q}{r}{$\bf-$1}} + C_{\five{s}{q}{p}{r}{$\bf-$1}} + C_{\five{s}{r}{p}{q}{$\bf-$1}} &=  C_{\four{\hspace{5 mm}s$\bf-$1}{\hspace{7 mm}p}{q}{r}} + C_{\four{s}{\hspace{3 mm}p$\bf-$1}{\hspace{2 mm}q}{r}} + C_{\four{s}{\hspace{-2 mm}p}{\hspace{-2 mm}q$\bf-$1}{r}} + C_{\four{s}{p}{\hspace{-6 mm}q}{\hspace{-4 mm}r$\bf-$1}} \,,\\
C_{\five{$\bf-$1}{s}{p}{q}{r}} &= E_{s-1} C_{\three{r}{p}{q}} +E_{s} C_{\three{r$\bf-$1}{p}{q}}\,,\\
C_{\four{0}{s}{t}{p}} &= E_{s} E_{t} E_{p} - C_{\three{s}{t}{p}}\,,\\
C_{\four{\hspace{2mm}$\bf{-}$1}{s}{t}{p}} &= E_{s-1} E_{t} E_{p} + E_{s} E_{t-1} E_{p} + E_{s} E_{t} E_{p-1}\,,\\
\label{eq:C3_simp}
 C_{\three{s}{t}{0}}&= E_{s} E_{t} - E_{s + t}\,,\\
\label{eq:C3_simp1}
 C_{\three{s}{t}{$\bf-$1}} &= E_{s-1} E_{t} +  E_{s} E_{t - 1} \,.
\end{align}
\end{subequations}
The last two already appeared in~\eqref{eq:CabcSpec}. The identities above can be derived most easily from the sum representation of the modular graph functions.}

\subsection{Laplace equations at low weights}

We now evaluate explicitly~\eqref{eq:LC6} for $C\!\!\raisebox{-0.3\height}{\six{s}{t}{p}{q}{w}{r}} $ for low weights $s+t+p+q+r+w$ starting from weight $6$.

\subsubsection{Laplace equation at weight $6$}

In order to illustrate the use of these equations, we re-derive the Laplace equation for the simplest non-trivial tetrahedral modular graph function that appears for weight $6$. From~\eqref{eq:LC6} one finds
\be
(\Laplace + 6 )C_{\six{1}{1}{1}{1}{1}{1}} = 12 \, C_{\six{$-$1}{1}{1}{1}{2}{2}} + 12 \, C_{\six{$-$1}{1}{2}{1}{2}{1}} - 24 C_{\six{0}{1}{1}{1}{2}{1}} - 24 \, C_{\six{0}{0}{2}{1}{2}{1}} + 12 \, C_{\six{0}{1}{2}{0}{2}{1}}  \, .
\ee
We simplify the right hand-side of the equation using equations \eqref{eq:simpT}:
\begin{subequations}
\begin{align}
C_{\six{$-$1}{1}{1}{1}{2}{2}}  + C_{\six{$-$1}{1}{2}{1}{2}{1}}  &= 2 E_{2}  \, C_{\three{1}{1}{2}} + E_{3}^2 - 2 C_{\three{1}{2}{3}} + 2 C_{\five{1}{1}{1}{2}{1}} - C_{\four{1}{1}{2}{2}} \,,\\
C_{\six{0}{0}{2}{1}{2}{1}} &= E_2 C_{\three{1}{1}{2}} - 2 C_{\three{1}{2}{3}}  \,,\\
C_{\six{0}{1}{2}{0}{2}{1}} &= E_6 - 2 E_3^2 + C_{\four{1}{1}{2}{2}} \,,\\
C_{\six{0}{1}{1}{1}{2}{1}} & = C_{\five{1}{1}{1}{2}{1}} - C_{\three{1}{2}{3}}\,.
\end{align}
\end{subequations}
Putting the results together we obtain
\begin{align}
\label{eq:L111111}
(\Laplace + 6 )C_{\six{1}{1}{1}{1}{1}{1}} = 48 C_{\three{1}{2}{3}} - 12 E_3^2 + 12 E_6 \, .
\end{align}
This equation was derived in this form first in \cite{Basu:2015ayg} and is relevant for determining the $D^{12}R^4$ correction at one loop.

\subsubsection{Laplace equation at weight $7$}

At weight $7$ there is only a single genuine tetrahedral modular graph function associated with the diagram
\begin{align}
\resizebox{.1\textwidth}{!}{
\begin{tikzpicture}
\draw [line width=1.5mm, gray, opacity = 0.3] (2,-1) -- (-2,-1) -- (0,2) -- (2,-1) -- (0,0) -- (0,2) ;
\draw [line width=1.5mm, gray, opacity = 0.3] (-2,-1) -- (0,0);
\greensB{-1}{0.5}{2}
\greensB{0}{1}{1}
\greensB{1}{0.5}{1}
\greensB{1}{-0.5}{1}
\greensB{-1}{-0.5}{1}
\greensB{0}{-1}{1}
\end{tikzpicture}}
\end{align}
Plugging this into the equation~\eqref{eq:LC6}, we find in a first instance
\begin{align}
(\Laplace + 6 )C_{\six{2}{1}{1}{1}{1}{1}} &=  8 C_{\six{$-$1}{1}{1}{1}{3}{2}}  + 8 C_{\six{$-$1}{1}{2}{1}{3}{1}}  + 2 C_{\six{$-$1}{1}{1}{2}{2}{2}}  + 2 C_{\six{$-$1}{1}{2}{2}{2}{1}} + 8 C_{\six{$-$1}{1}{2}{1}{2}{2}} - 8 C_{\six{0}{0}{2}{1}{3}{1}}\nn\\
&\quad - 8 C_{\six{0}{0}{2}{1}{2}{2}}  - 8 C_{\six{0}{0}{3}{1}{2}{1}} - 8 C_{\six{0}{1}{1}{1}{3}{1}} - 10 C_{\six{0}{1}{1}{1}{2}{2}}   - 4 C_{\six{0}{1}{1}{2}{2}{1}}  + 8 C_{\six{0}{1}{2}{0}{3}{1}} \nn\\
&\quad + 4 C_{\six{0}{1}{2}{0}{2}{2}}  -6 C_{\six{0}{1}{2}{1}{2}{1}} - 4 C_{\six{0}{1}{2}{1}{1}{2}}  \, .
\end{align}
For the simplifications we use again~\eqref{eq:simpT} and there are many cancellations.
Combining all the terms together the final Laplace equation at weight $7$ is
\begin{align}
(\Laplace + 6 )C_{\six{2}{1}{1}{1}{1}{1}} =  - 2 C_{\five{1}{2}{1}{2}{1}} + 30 C_{\three{1}{2}{4}} + 18 C_{\three{1}{3}{3}} + 8 C_{\three{2}{2}{3}} - 12 E_3 E_4 + 12 E_7\,.
\end{align}

\subsection{Generating function, its Laplace equation and \texorpdfstring{$\mf{sl}(3)$}{sl(3,R)}}

For understanding more general properties of the spectrum of the Laplacian on tetrahedral modular graph functions, it is useful to consider a generating function, similar to the one introduced in~\cite{DHoker:2015foa}. For the tetrahedral graphs considered here we write it as
\begin{align}
\cW_{\six{\hspace{-2 mm}t{\scriptscriptstyle_1}}{t{\scriptscriptstyle_2}}{t{\scriptscriptstyle_3}}{t{\scriptscriptstyle_4}}{t{\scriptscriptstyle_6}}{t{\scriptscriptstyle_5}}} &= \sum_{s,t,p,q,w,r = 1}^{\infty} t^{s-1}_1 t^{t-1}_2 t^{p-1}_3 t^{q-1}_4 t^{r-1}_5 t^{w-1}_6 C_{\six{\hspace{-2 mm}s}{t}{p}{q}{w}{r}} \, .
\end{align}
In terms of the lattice sum this can be thought of as considering massive propagators between the vertices $z_i$
\be
\cW(t_i, \tau, \bar{\tau}) = \sump_{(m_{i},n_{i}) \in \mathbb{Z}^{2}} \delta^{z_i}_{m 0} \delta^{z_{i}}_{n 0} \prod_{j = 1}^{6}  \frac{\tau_{2}}{\pi |m_{j} + n_{j} \tau|^{2} - t_r \tau_2}\,.
\ee

We will now determine the action of the Laplace operator on $\cW$ from the Laplace equation~\eqref{eq:LC6}. We begin with the `eigenvalue' $\omega$ shown in~\eqref{eq:om}.  The left-hand side of the Poisson equations can be expressed using the relation
\be
\sum_{s,t,p,q,w,r = 1}^{\infty} t^{s-1}_1 t^{t-1}_2 t^{p-1}_3 t^{q-1}_4 t^{r-1}_5 t^{w-1}_6 s(s-1) C_{\six{\hspace{-2 mm}s}{t}{p}{q}{w}{r}} = t_1 \partial^2_{1} (t_1 \cW_{\six{\hspace{-2 mm}t{\scriptscriptstyle_1}}{t{\scriptscriptstyle_2}}{t{\scriptscriptstyle_3}}{t{\scriptscriptstyle_4}}{t{\scriptscriptstyle_6}}{t{\scriptscriptstyle_5}}} )  \, .
\ee
Here we use the notation $\partial_i\equiv \partial/\partial_{t_i}$ as a short-hand. This part can be rewritten for all legs as $\left(\sum_{i = 1}^{6} t_i \partial^2_{i} t_i \right)\cW$. The mixed terms in $\omega$ of~\eqref{eq:om} can be written in terms of $t$-derivatives as 
\be
\sum_{s,t,p,q,w,r = 1}^{\infty} t^{s-1}_1 t^{t-1}_2 t^{p-1}_3 t^{q-1}_4 t^{r-1}_5 t^{w-1}_6 w p C_{\six{\hspace{-2 mm}s}{t}{p}{q}{w}{r}} = \partial_{3} \partial_{6} (t_3 t_6 \cW_{\six{\hspace{-2 mm}t{\scriptscriptstyle_1}}{t{\scriptscriptstyle_2}}{t{\scriptscriptstyle_3}}{t{\scriptscriptstyle_4}}{t{\scriptscriptstyle_6}}{t{\scriptscriptstyle_5}}} )  \, .
\ee
We therefore deduce that
\begin{align}
\label{eq:omgen}
&\sum_{s,t,p,q,w,r = 1}^{\infty} t^{s-1}_1 t^{t-1}_2 t^{p-1}_3 t^{q-1}_4 t^{r-1}_5 t^{w-1}_6 \omega C_{\six{s}{t}{p}{q}{w}{r}} 
=\left(\sum_{i=1}^6 t_i \partial_i^2 t_i -2 \sum_{i=1}^3 \partial_i \partial_{i+3} t_i t_{i+3}\right)  \cW_{\six{\hspace{-2 mm}t{\scriptscriptstyle_1}}{t{\scriptscriptstyle_2}}{t{\scriptscriptstyle_3}}{t{\scriptscriptstyle_4}}{t{\scriptscriptstyle_6}}{t{\scriptscriptstyle_5}}} \nn\\
&\hspace{20mm}= \left(\cD^2-\cD-6- 2 \sum_{<i,j>} t_i t_j \partial_i \partial_j - 4 \sum_{i=1}^3 t_i t_{i+3} \partial_i \partial_{i+3} \right) \cW_{\six{\hspace{-2 mm}t{\scriptscriptstyle_1}}{t{\scriptscriptstyle_2}}{t{\scriptscriptstyle_3}}{t{\scriptscriptstyle_4}}{t{\scriptscriptstyle_6}}{t{\scriptscriptstyle_5}}} \,,
\end{align}
where $\omega$ on the left-hand side is given by~\eqref{eq:om}. We have separated the sum over pairs of edges into the $12$ adjacent pairs $<i,j>$ and the three opposite pairs $(i,i+3)$ for $i=1,2,3$. The differential operator 
\begin{align}
\cD = \sum_{i=1}^6 t_i \partial_i
\end{align}
measures the degree of homogeneous polynomials in the $t_i$. 

Next we analyse the inhomogeneous terms on the right-hand side of the Laplace equation~\eqref{eq:LC6}. We will use again the short-hand~\eqref{eq:6short} to indicate a number of propagators different from the standard one in $C\!\!\raisebox{-0.3\height}{\six{s}{t}{p}{q}{w}{r}}$. As a rule of thumb, any shifted index will be associated with a shifted power on the corresponding variable $t_i$ in the generating function. Shifting the summation back to the standard range will produce `boundary terms' where some of the edges have the special values that also appear in~\eqref{eq:simpT}. Generally, only the edges with negative shifts will contribute to these boundary terms; the positive $+1$ shifts are innocuous as they only appear for the edges whose power also multiplies the corresponding contribution.

Let us consider as an example the first term on the right-hand side of~\eqref{eq:LC6} that contains an adjacent pair of lines:
\begin{align}
\label{eq:RHSgen}
&\sum_{s,t,p,q,w,r = 1}^{\infty} t^{s-1}_1 t^{t-1}_2 t^{p-1}_3 t^{q-1}_4 t^{r-1}_5 t^{w-1}_6 st C_{\six{1}{\hspace{-2 mm}$\bf  -$1}{}{}{}{}} 
=  \partial_1 \partial_2\!\!\left[\! t_1 t_2 \sum_{\substack{p,q,r,w = 1\\t = 0 \\ s= 2}}  t^{s-2}_1 t^{t}_2 t^{p-1}_3 t^{q-1}_4 t^{r-1}_5 t^{w-1}_6C_{\six{}{}{}{}{}{}}\right] \nn\\
&\hspace{20mm}=  \partial_1 \partial_2 \left[t_2^2 \sum_{\substack{s,p,q,r,w = 1\\t = 0 }}  t^{s-1}_1 t^{t-1}_2 t^{p-1}_3 t^{q-1}_4 t^{r-1}_5 t^{w-1}_6C_{\six{}{}{}{}{}{}} \right]\nn\\
&\hspace{20mm}=  \partial_1 \partial_2 \left[t_2^2 \cW_{\six{\hspace{-2 mm}t{\scriptscriptstyle_1}}{t{\scriptscriptstyle_2}}{t{\scriptscriptstyle_3}}{t{\scriptscriptstyle_4}}{t{\scriptscriptstyle_6}}{t{\scriptscriptstyle_5}}}
  + t_2  \cW_{\six{\hspace{-2 mm}t{\scriptscriptstyle_1}}{0}{t{\scriptscriptstyle_3}}{t{\scriptscriptstyle_4}}{t{\scriptscriptstyle_6}}{t{\scriptscriptstyle_5}}}\right]\,.
\end{align}
The last term comes from the $t=0$ term in the sum and we have introduced the notation
\begin{align}
\cW_{\six{\hspace{-2 mm}t{\scriptscriptstyle_1}}{0}{t{\scriptscriptstyle_3}}{t{\scriptscriptstyle_4}}{t{\scriptscriptstyle_6}}{t{\scriptscriptstyle_5}}} =  \sum_{s,p,q,r,w = 1}  t^{s-1}_1  t^{p-1}_3 t^{q-1}_4 t^{r-1}_5 t^{w-1}_6C_{\six{}{0}{}{}{}{}} =\cW_{\six{\hspace{- 1 mm}1}{2}{3}{4}{6}{5}} \Big|_{2\to0}\
\end{align}
for the generating function of degenerate tetrahedral graph functions. The function $C\!\!\raisebox{-0.3\height}{\six{}{0}{}{}{}{}}$ could in principle be simplified using~\eqref{eq:simpT}, but it is more compact to leave it in this form. We see that a term on the right-hand side of~\eqref{eq:LC6} contributes both to a differential operator acting on $\cW\!\!\raisebox{-0.3\height}{\six{\hspace{-2 mm}t{\scriptscriptstyle_1}}{t{\scriptscriptstyle_2}}{t{\scriptscriptstyle_3}}{t{\scriptscriptstyle_4}}{t{\scriptscriptstyle_6}}{t{\scriptscriptstyle_5}}}$ and to degenerate boundary terms. 

Manipulations similar to~\eqref{eq:RHSgen} can be performed for all the adjacent lines $\langle i,j\rangle$ and opposite lines in~\eqref{eq:LC6}. Summing up all the contributions then gives
\begin{align}
\label{eq:LG6}
\bigl( \Laplace - \cD^2 + \cD + 6\bigr) \cW_{\six{\hspace{- 1 mm}1}{2}{3}{4}{6}{5}}  &=  \sum_{V_{ijk}} \bigl( t^2_i + t^2_j + t^2_k - 2 t_i  t_j - 2 t_i t_k - 2 t_j t_k \bigr)  (\partial_{j} \partial_{k}  + \partial_{i} \partial_{j} + \partial_{i} \partial_{k}) \cW_{\six{\hspace{- 1 mm}1}{2}{3}{4}{6}{5}}  \nn\\ 
&\quad\quad+ \sum_{i = 1}^{3} ( (t_{i+1} -  t_{i+2}  + t_{i+4}  - t_{i + 5})^2 - 4 t_{i} t_{i + 3} )\partial_{i} \partial_{i+3} \cW_{\six{\hspace{- 1 mm}1}{2}{3}{4}{6}{5}} + \cR \, ,
\end{align}
where the two sums arise from the adjacent lines $\langle i,j\rangle$ coming together at a vertex $V_{ijk}$ and the three pairs of opposite lines. We have also moved some of the terms in~\eqref{eq:omgen} to the right. The term $\cR$ contains all the contributions from degenerate diagrams and is given explicitly by
\begin{align}
\cR&= - \sum_{i=1}^6 \sum_{V_{pqi}} \left( \partial_p + \partial_q +(2t_p + 2t_q - t_i) \partial_p \partial_q\right)  \cW_{\six{\hspace{- 1 mm}1}{2}{3}{4}{6}{5}} \Big|_{i\to0}\nn\\
&\quad+ \sum_{i=1}^6 \sum_{\substack{p=1\\p\notin\{i,i+3\}}}^3 \left(t_i +2t_{i+3}-2t_{p+1}-2t_{p+4}\right)\partial_p \partial_{p+3} \cW_{\six{\hspace{- 1 mm}1}{2}{3}{4}{6}{5}} \Big|_{i\to0}\nn\\
&\quad+ \sum_{i=1}^6 \sum_{\Delta_{ipq}\neq \Delta_{irs}}^6 \left(\partial_p +\partial_q\right)\left(\partial_r+\partial_s\right) \cW_{\six{\hspace{- 1 mm}1}{2}{3}{4}{6}{5}} \Big|_{i\to -1}\nn\\
&\quad+\!\bigg[ 2 \partial_{1} \partial_{4} \Bigl( \cW_{\six{\hspace{- 2 mm}1}{0}{3}{4}{6}{0}} +\cW_{\six{\hspace{- 2 mm}1}{2}{0}{4}{0}{5}} -\cW_{\six{\hspace{- 2 mm}1}{0}{0}{4}{6}{5}} 
- \cW_{\six{\hspace{- 2 mm}1}{0}{3}{4}{0}{5}} - \cW_{\six{\hspace{- 2 mm}1}{2}{0}{4}{6}{0}}  - \cW_{\six{\hspace{- 2 mm}1}{2}{3}{4}{0}{0}} \Bigr) \nn\\
&\hspace{20mm}+ \textrm{two other pairs of opposite lines} \biggr] 
\end{align}
(Indices are to be understood modulo $6$.) The four terms are almost simpler to describe in words: The first term is a sum over the six edges $i$ with $p$ and $q$ joining line $i$ at a vertex; so for $i=1$ it would be the two cases $(p,q)\in \{ (2,3), (5,6)\}$ because edge $4$ is opposite of edge $1$. The second term is also an outer sum over the edges $i$ and the inner sum produces the two pairs of opposite edges not containing $i$; for $i=1$ it would be $(t_1+2t_4-2t_3-2t_6)\partial_2\partial_5+ (t_1+2t_4-2t_2-2t_5)\partial_3\partial_6$ . The third term is also a sum over all the edges $i$ and the inner sum contains the two triangles that can be formed containing the edge $i$; for $i=1$ this means $(\partial_3+\partial_5)(\partial_2+\partial_6)$. The very last term comes from the three pairs of opposite edges and has two degenerations in the generating function with sign distributions depending on whether the degenerations are on opposite or adjacent edges. 
In the case considered in~\cite{DHoker:2015foa}, all boundary terms could be simplified to Eisenstein series or products thereof by virtue of~\eqref{eq:CabcSpec}; here the source terms are of a more complicated nature but still simpler compared to the full tetrahedral function. This can be seen in the examples above. 

As is shown in appendix~\ref{app:SL3C}, the Laplace equation~\eqref{eq:LG6} for the generating function can be rewritten using the quadratic Casimir of the split real $\mf{sl}(3)$. Upon inserting the Casimir 
\begin{align}
\mathfrak{C}^2 &= \frac43 \mathfrak{D}^2 + 2 \mathfrak{D} +
\sum_{V_{ijk}} (t_i^2+t_j^2+t_k^2 - 2t_i t_j -2 t_j t_k -2t_k t_i) (\partial_i \partial_j +\partial_j \partial_k +\partial_k\partial_i)\nn\\
& \quad\quad\quad + \sum_{i=1}^3 ((t_{i+1} - t_{i+2} -t_{i+4} +t_{i+5})^2 -4 t_i t_{i+3}) \partial_i \partial_{i+3}\,.
\end{align}
that is derived in appendix~\ref{app:SL3C}, we obtain
\begin{align}
\label{eq:LapCas}
\left( \Laplace -\mf{C}^2 +\frac13\left(\cD+3\right)\left(\cD+6\right)\right) \cW = \cR\,,
\end{align}
where we have suppressed all dependence on the variables $\tau$ and $t_i$ of the generating function $\cW$ and the remainder $\cR$. Solving the spectral problem means finding the spectrum of the operator
\begin{align}
\label{eq:L2C2}
\mf{L}^2 = \mf{C}^2 -\frac13\left(\cD+3\right)\left(\cD+6\right)\,.
\end{align}
We note that the occurrence of $\mf{sl}(3)$ is slightly misleading, there is no actual $\mf{sl}(3)$ symmetry of the spectrum; what we will be interested in is the number of $\mf{S}_4$ singlets in representations of $\mf{sl}(3)$. This situation is generalisation of the case discussed in appendix~\ref{app:2loop} for the sunset graph underlying the functions $C_{(s,t,p)}$.

\subsection{General spectral considerations}

We now try to find a basis of $\mf{C}^2$-eigenfunctions of homogeneous polynomials in the six $t_i$ that transform under $\mathfrak{S}_4$ in such a way that the polynomials are invariant. The action of $\mf{S}_4$ is induced from that of $\mf{sl}(3)$ mentioned above and exhibited in appendix~\ref{app:SL3C}. That is, we are looking for functions that satisfy
\begin{align}
\mathfrak{D} \mathcal{W}_{w,p_1,p_2} & = (w-6) \mathcal{W}_{w,p_1,p_2}\,,\nn\\
\mathfrak{C}^2 \mathcal{W}_{w,p_1,p_2} & = \frac13 (p_1^2+p_1p_2+p_2^2+3p_1+3p_2)\mathcal{W}_{w,p_1,p_2} \equiv\lambda_{p_1,p_2} \mathcal{W}_{w,p_1,p_2}
\end{align}
and are $\mf{S}_4$-invariant. We here are using the $\mf{sl}(3)$ quadratic Casimir operator $\mf{C}^2$ that was defined in~\eqref{eq:CasSL3} that is normalised such that when acting on an $\mf{sl}(3)$ representation with Dynkin labels $[p_1,p_2]$ it has eigenvalue $\lambda_{p_1,p_2} =  \frac13 (p_1^2+p_1p_2+p_2^2+3p_1+3p_2)$. We note that the dimension of the $\mf{sl}(3)$ representation with labels $[p_1,p_2]$ is given by
\begin{align}
\dim\, [p_1,p_2] = \frac{(p_1+p_2+2)(p_1+1)(p_2+1)}{2}\,.
\end{align}

The eigenvalue $k$ of the scaling operator $\mf{D}$ is related to the weight $w=\sum_i s_i$ of $C\!\!\raisebox{-0.3\height}{\six{s_1}{s_2}{s_3}{s_4}{s_6}{s_5}}$ by
\begin{align}
k = w-6\,.
\end{align}
The value $k$ corresponds to the degree of the homogeneous polynomial in the $t_i$.

The operators $\mf{D}$ and $\mf{C}^2$ commute and the eigenvalue of the modular Laplacian is then 
\begin{align}
\label{eq:CD}
\left( \Delta -\lambda_{p_1,p_2} +\frac13w(w-3)\right) \mathcal{W}_{w,p_1,p_2} = \mathcal{R}_{w,p_1,p_2}
\end{align}
according to~\eqref{eq:L2C2}. We note that $\mf{D}$ and $\mf{C}^2$ do not form a complete set of commuting semi-simple operators. There are still degeneracies in the eigenspace labelled by $(w,p_1,p_2)$. The form of the right-hand side above depends on which particular eigenfunction in the $(w,p_1,p_2)$ we are considering.

Mimicking the analysis of the two-loop sunset graph in appendix~\ref{app:2loop}, we need to first identify the correct representations of $\mf{sl}(3)$. The representation on six variables $t_i$ is the irreducible six-dimensional representation that we choose as $[2,0]$ by some labelling convention for the nodes of the $A_2$ Dynkin diagram. For homogeneous polynomials of degree $k$ we need to form its symmetric tensor product series. We first compute the total number of $\mf{S}_4$-invariant functions for a given degree $k=w-6$ of the polynomial. This can be done by applying Molien's theorem to the six-dimensional representation of $\mf{S}_4$ and gives the following generating function
\begin{align}
\frac{1-q+q^2+q^4+q^6-q^7+q^8}{(1-q)^6 (1+q)^2 \left(1+q^2\right) \left(1+q+q^2\right)^2}
\end{align}
that is also documented as series \href{https://oeis.org/A003082}{A003082} in the OEIS~\cite{OEIS}. From this one can construct the total number of $\mf{S}_4$ singlets at a given order
\begin{align}
\label{eq:tetspec}
\arraycolsep=2mm
\begin{array}{c|cccccccc}
k & 0 & 1 & 2 & 3 & 4 & 5 & 6 & 7\\\hline
\mf{S}_4 \textrm{ singlets in $S^k([2,0])$} & 1 & 1 & 3 & 6 & 11 & 18 & 32 & 48
\end{array}
\end{align}
As $w=k+6$, we recognise the single tetrahedral modular graph function at weight $6$ and the single tetrahedral modular graph function at weight $7$ discussed above.

In order to separate the total number of eigenfunctions at weight $w=k+6$ into the Casimir eigenspaces of the representation $[p_1,p_2]$ we need to consider the $\mf{sl}(3)$ representations occurring in the $k$-th symmetric tensor power of the six-dimensional representation $[2,0]$ of the $t_i$ variables. This is given by 
\begin{align}
\label{eq:Sk20}
S^k ([2,0]) = \bigoplus_{i=0}^{\lfloor\frac{k}{3}\rfloor} \bigoplus_{j=0}^{\lfloor\frac{i}{2}\rfloor} [2k-6i,6j] 
\oplus \bigoplus_{i=0}^{\lfloor\frac{k-2}{3}\rfloor}  \bigoplus_{j=0}^{\lfloor\frac{i}{2}\rfloor}[2k-6i-4,6j+2]
\oplus \bigoplus_{i=0}^{\lfloor\frac{k-1}{3}\rfloor} \bigoplus_{j=0}^{\lfloor\frac{i-1}{2}\rfloor}[2k-6i-2,6j+4] .
\end{align}
All these representations of $\mf{sl}(3)$ occur with multiplicity one. The only degeneracies in the Casimir eigenvalues arise for representations related by the outer automorphism of $\mf{sl}(3)$, \textit{i.e.}, only $[p_1,p_2]$ and $[p_2,p_1]$ have the same Casimir eigenvalue, otherwise all Casimir values are different.

Unfortunately, compared to the two-loop case of appendix~\ref{app:2loop}, we do not have a direct description of all $[p_1,p_2]$ as symmetric powers of some simple representation. A notable exception is again given by the symmetric powers of the fundamental (and anti-fundamental) representation:
\begin{align}
S^p([1,0])= [p,0] \quad\quad \textrm{($S^p([0,1]) = [0,p]$)}\,.
\end{align}
As a representation of $\mf{S}_4$ the three-dimensional fundamental representation of $\mf{sl}(3)$ is the standard representation and one can choose as generators for example the matrices
\begin{align}
\begin{pmatrix}
0 & 0 & 1\\
0 & 1 & 0\\
-1 & 1 &0 
\end{pmatrix}
\,,\quad
\begin{pmatrix}
0 & 0 & -1\\
0 & -1 & 0\\
-1 & 0 &0 
\end{pmatrix}
\,,\quad
\begin{pmatrix}
0 & 0 & 1\\
1 & -1 & 1\\
1 & 0 &0 
\end{pmatrix}\,.
\end{align}

Molien's theorem gives the number of $\mathfrak{S}_4$ invariants in such representations as being generated by
\begin{align}
\frac{1-q^3+q^6}{(1-q)^3(1+q)^2(1+q^2)(1+q+q^2)}\,.
\end{align}
For low  $p$ one has for the number of $\mf{S}_4$ singlets in $[p,0]$ (or equivalently $[0,p]$)
\begin{align}
\arraycolsep=2mm
\begin{array}{c|ccccccc}
p & 0 & 1 & 2 & 3 & 4 & 5 & 6\\\hline
\mf{S}_4 \textrm{ singlets in $[p,0]$} & 1 & 0 & 1 & 0 & 2 & 0&  3 
\end{array}
\end{align}

\noindent (The series starts to be non-zero for odd $p$ soon after this.) At \textit{even} $p$ this series agrees with 
\begin{align}
\frac{1}{(1-q^2)^3(1+2q^2+2q^4+q^6)}= \frac{1}{(1-q^2)(1-q^4)(1-q^6)}
\end{align}
as can be checked using a double extension of $\mathfrak{S}_4$. The even $p$ values are the only ones of interest to us, so we might as well work with this simpler series that is identical to the two-loop case.  

For more general $[p_1,p_2]$ it is not quite clear how to get the right generating series. In table~\ref{tab:S4sing}, we present the number of $\mf{S}_4$ invariants for small even $p_1\geq p_2$. This inequality is sufficient due to the outer automorphism of $\mf{sl}(3)$ and only even values of the $p_i$ can arise in~\eqref{eq:Sk20}.
\begin{table}[t!]
\centering
\caption{\label{tab:S4sing}\sl Number of $\mf{S}_4$ singlets in various representations of $\mf{sl}(3)$ together with their Casimir eigenvalues and the weight of the tetrahedral modular graph function for which they arise first. The list is ordered by the dimension of the representation and complete up to dimension $162$.}

\vspace{2mm}
\begin{tabular}{c||c|c|c|c}
$\mf{sl}(3)$ & first occurrence & Casimir & $\dim\, [p_1,p_2]$ & Number of \\
representation & at weight $w=k+6$ & value  $\mf{C}^2$ & & $\mf{S}_4$ singlets \\\hline
$[0,0]$ & $6$ & $0$ & $1$ & $1$\\
$[2,0]$ & $7$ & $\frac{10}3$ & $6$ & $1$ \\
$[4,0]$ & $8$ & $\frac{28}3$ & $15$ & $2$ \\
$[2,2]$ & $9$ & $8$ &$27$ & $2$ \\
$[6,0]$ & $9$ & $18$ & $28$ & $3$ \\
$[8,0]$ & $10$  & $\frac{88}{3}$ & $45$ & $4$ \\
$[4,2]$ & $10$ & $\frac{46}{3}$ & $60$ & $4$ \\
$[10,0]$ & $11$  & $\frac{130}{3}$ & $66$ & $5$ \\
$[12,0]$ & $12$  & $60$ & $91$ & $7$ \\
$[6,2]$ & $11$  & $\frac{106}{3}$ & $105$ & $6$ \\
$[4,4]$ & $12$ & $24$ & $125$ & $7$ \\
$[8,2]$ &  $12$   & $38$ & $162$ & $9$ 
\end{tabular}
\end{table}

The eigenvalue of the modular Laplacian at a given weight $w$ then has to be calculated using~\eqref{eq:CD}.

\subsection{Eigenfunctions and Laplace equations at low weight}

In this section, we give some more examples of eigenvalues and eigenfunctions of the modular Laplacian acting on tetrahedral modular graph functions. We stress that we are using the terms `eigenvalues' and `eigenfunctions' loosely as the corresponding Laplace equations are typically inhomogeneous but the right-hand side source is of lower complexity.

The explicit eigenfunctions of the modular Laplacian at low weights can be constructed using a basis of $\mf{S}_4$-invariant homogeneous polynomials of degree $k$. We list as examples the linear and quadratic invariant homogeneous polynomials:
\begin{align}
k=1:&& p_1(t) &= t_1+t_2+t_3+t_4+t_5+t_6\,,\nn\\
k=2: && p_2^{(1)}(t)&=t_1^2+t_2^2+t_3^2+t_4^2+t_5^2+t_6^2\,,\nn\\ 
&& p_2^{(2)}(t)&=t_1 t_2 + t_1 t_3 + t_2 t_3 + t_2 t_4 + t_3 t_4 + t_1 t_5 + t_2 t_5 + t_4 t_5 + 
 t_1 t_6 + t_3 t_6 + t_4 t_6 + t_5 t_6\,,\nn\\
 && p_2^{(3)}(t)&=t_1t_4+t_3t_5 + t_2t_6\,.
\end{align}
Similar bases of $\mf{S}_4$ invariant homogeneous polynomials can be generated at any degree easily. The procedure for finding explicit eigenfunctions of $\mf{L}^2$ is then to first diagonalise the action of the $\mf{sl}(3)$ Casimir $\mf{C}^2$ on the polynomials and then convert this to inhomogeneous Laplace equations for combinations of tetrahedral modular graph functions.

\subsubsection{Weight $7$} 

The linear polynomial $p_1$ is an eigenfunction of the $\mf{sl}(3)$ Casimir $\mf{C}^2$ given in~\eqref{eq:CasSL3} according to
\begin{align}
\left(\mf{C}^2 - \frac{10}{3}\right) p_1 = 0\,.
\end{align}
The corresponding tetrahedral modular graph function is $C\raisebox{-0.5\height}{\six{2}{1}{1}{1}{1}{1}}$ with Laplace equation given already in~\eqref{eq:L111111}.

\subsubsection{Weight $8$} 

For $k=2$ the following are explicit eigenfunctions of the Casimir operator~\eqref{eq:CasSL3} in the normalisation given there:
\begin{subequations}
\label{eq:deg2ef}
\begin{align}
\left(\mf{C}^2- \frac{10}3 \right)\left( p_2^{(1)} -  p_2^{(2)} \right) &=0\,,\\
\left(\mf{C}^2  -\frac{28}3 \right)  p_2^{(1)} &=0\,,\\
\left(\mf{C}^2- \frac{28}3 \right)\left( p_2^{(2)} -  3p_2^{(3)} \right) &=0\,.
\end{align}
\end{subequations}
This is the first time a degeneracy arises in the spectrum and we have chosen some particular simple basis. These $\mf{S}_4$ invariant eigenfunctions~\eqref{eq:deg2ef} of the $\mf{sl}(3)$ Casimir $\mf{C}^2$ can be translated into combinations of tetrahedral modular graph functions as follows 
\begin{subequations}
\label{eq:Inhdeg2}
\begin{align}
\left(\Delta +4 \right) C_{\six{2}{1}{1}{2}{1}{1}} &= -16 C_{\six{-1}{3}{3}{1}{1}{1}}+32 C_{\five{1}{2}{1}{3}{1}}-24C_{\five{1}{2}{2}{2}{1}}+8C_{\five{1}{2}{1}{2}{2}} + 2 C_{\four{2}{2}{2}{2}} \nn\\
&\quad -8C_{\four{1}{2}{2}{3}} +88 C_{\three{1}{3}{4}} -12 C_{\three{2}{2}{4}} -4C_{\three{2}{3}{3}}\nn\\
&\quad +8 E_3 C_{\three{1}{2}{2}} -14 E_4^2 +20 E_8\,,\\
\left(\Delta +4 \right) \left(C_{\six{3}{1}{1}{1}{1}{1}}+C_{\six{2}{2}{1}{1}{1}{1}}\right) &=  4 C_{\six{-1}{3}{3}{1}{1}{1}} -16 C_{\five{1}{2}{1}{3}{1}}+2C_{\five{1}{2}{2}{2}{1}}-2C_{\five{1}{2}{1}{2}{2}} -\frac12 C_{\four{2}{2}{2}{2}} \nn\\
&\quad +2C_{\four{1}{2}{2}{3}} +60 C_{\three{1}{2}{5}}+40 C_{\three{1}{3}{4}} +18 C_{\three{2}{2}{4}} +16 C_{\three{2}{3}{3}}\nn\\
&\quad -2 E_3 C_{\three{1}{2}{2}} -\frac{11}{2} E_4^2 -20E_3 E_5 +24 E_8\,,\\
\left(\Delta +10 \right) \left(C_{\six{2}{1}{1}{2}{1}{1}}+3C_{\six{2}{2}{1}{1}{1}{1}}\right) &=  -4 C_{\six{-1}{3}{3}{1}{1}{1}} +8 C_{\five{1}{2}{1}{3}{1}}-6C_{\five{1}{2}{2}{2}{1}}-4C_{\five{1}{2}{1}{2}{2}} +\frac12 C_{\four{2}{2}{2}{2}} \nn\\
&\quad -2C_{\four{1}{2}{2}{3}} +36 C_{\three{1}{2}{5}}+124 C_{\three{1}{3}{4}} +42 C_{\three{2}{2}{4}} +32 C_{\three{2}{3}{3}}\nn\\
&\quad +2 E_3 C_{\three{1}{2}{2}} -\frac{61}{2} E_4^2 -12E_3 E_5 +44 E_8\,.
\end{align}
\end{subequations}
As can be seen, all right-hand sides contain the function $C\raisebox{-0.5\height}{\six{-1}{3}{3}{1}{1}{1}}$. This function cannot be reduced by means of the simplification rules given in~\eqref{eq:simpT}. However, we expect there to be an additional simplification rule that we have not been able to derive and that would simplify this function.

\subsubsection{Weight $9$} 

For $k=3$ one has in total six eigenfunctions according to~\eqref{eq:tetspec}. These separate into the $\mf{L}^2$ eigenvalues $-18$, $-10$ and $0$ with degeneracies $1$, $2$ and $3$, respectively. For $\mf{L}^2=-18$ the eigenfunction is given by
\begin{align}
\mf{L}^2=-18:\quad& C_{\six{2}{2}{1}{1}{2}{1}}  +3C_{\six{2}{2}{1}{2}{1}{1}}  \,.
\end{align}
For $\mf{L}^2=-10$ one has the following basis of two eigenfunctions
\begin{align}
\mf{L}^2=-10:\quad& C_{\six{2}{2}{1}{1}{2}{1}}  +2C_{\six{2}{2}{1}{2}{1}{1}}   +2C_{\six{3}{1}{1}{2}{1}{1}}   \,, \nn\\
&C_{\six{2}{2}{2}{1}{1}{1}}  +C_{\six{2}{2}{1}{1}{2}{1}}   +3C_{\six{2}{2}{1}{2}{1}{1}}+4C_{\six{3}{2}{1}{1}{1}{1}}  +2C_{\six{3}{1}{1}{2}{1}{1}} \,.
\end{align}
For $\mf{L}^2=0$ one has the follow three independent eigenfunctions
\begin{align}
\mf{L}^2=-10:\quad&  C_{\six{2}{2}{2}{1}{1}{1}}  +C_{\six{2}{2}{1}{2}{1}{1}}   +4C_{\six{3}{2}{1}{1}{1}{1}}+5C_{\six{4}{1}{1}{1}{1}{1}}\,,\nn\\
& -C_{\six{2}{2}{1}{1}{2}{1}}   +3C_{\six{2}{2}{1}{2}{1}{1}}+3C_{\six{3}{1}{1}{2}{1}{1}}\,,\nn\\
&  3C_{\six{2}{2}{2}{1}{1}{1}}  -2C_{\six{2}{2}{1}{1}{2}{1}} -6C_{\six{2}{2}{1}{2}{1}{1}}   -18C_{\six{3}{2}{1}{1}{1}{1}}-9C_{\six{3}{1}{1}{2}{1}{1}}\,.
\end{align}
We do not spell out the right-hand sides of the inhomogeneous Laplace equations as they are rather involved but note that, similar to~\eqref{eq:Inhdeg2} they can involve tetrahedral modular graph functions on the right-hand side with where one edge has value $-1$. Such terms possibly simplify.

\subsubsection{Weights $10$, $11$ and $12$}

For weights $10$, $11$ and $12$ we only present table~\ref{tab:higherwt} of the degeneracies of the eigenvalues of $\mf{L}^2$ and do not list the explicit eigenfunctions as they become rather involved.
\begin{table}[t]
\caption{\label{tab:higherwt}\sl Eigenvalues and degeneracies of the modular Laplacian acting on tetrahedral modular graph functions of weights $10$, $11$ and $12$.}
\vspace{2mm}

\centering
\begin{tabular}{c||c|c|c}
weight &  $\mf{L}^2$ eigenvalue & $\mf{sl}(3)$ rep. & $\mf{S}_4$ singlets \\\hline\hline
$10$ & $-20$ & $[2,0]$ & $1$ \\
& $-14$ & $[0,4]$ & $2$ \\
& $-8$ & $[8,0]$ & $4$ \\
& $6$ & $[4,2]$ & $4$\\\hline
$11$ & $-26$  & $[0,2]$ & $1$ \\
& $-20$  & $[4,0]$ & $2$ \\
& $-14$  & $[2,4]$ & $4$ \\
&  $-4$ & $[10,0]$ & $5$ \\
&  $14$ & $[6,2]$ & $6$ \\\hline
$12$ & $-36$  & $[0,0]$ & $1$\\
& $-28$   & $[2,2]$ & $2$ \\
& $-18$   & $[6,0]\oplus[0,6]$ & $6$ \\
&  $-12$  & $[4,4]$ & $7$ \\
&  $2$  & $[8,2]$ & $9$ \\
&  $24$  & $[12,0]$ & $7$ 
\end{tabular}
\end{table}

The numbers in table~\ref{tab:higherwt} can also be derived from~\eqref{eq:Sk20},~\eqref{eq:L2C2} and table~\ref{tab:S4sing}. We have additionally determined the corresponding eigenfunctions and checked that their inhomogeneous Laplace equations contain only less complex modular graph functions or tetrahedral modular graph functions that have one edge with value $-1$. 

\subsubsection{More Laplace equations}

In this section we present some additional Laplace equations, where the Laplacian is not diagonalised as in the previous examples but instead the combinations are chosen such that there are no functions with a value $-1$ on any edge remaining on the right-hand side. These together with the previous ones could be useful for finding the integrated versions of the corresponding amplitudes. In general, there remain tetrahedral modular graph functions with a similar complexity on the right-hand side.

For weight $w=8$ one has 
\begin{align}
\Laplace \biggr(  C_{\six{2}{1}{1}{2}{1}{1}} + 4 C_{\six{2}{2}{1}{1}{1}{1}} \biggr) &= -12 C_{\six{2}{1}{1}{2}{1}{1}} - 40 C_{\six{2}{2}{1}{1}{1}{1}} + 48 C_{\three{1}{2}{5}}  + 136  C_{\three{1}{3}{4}} + 60  C_{\three{2}{2}{4}} \nn\\
&\quad+44 C_{\three{2}{3}{3}}  -8 C_{\five{1}{2}{1}{2}{2}} -36 E^2_{4} -16 E_{3} E_{5} + 52  E_{8} \,.
\end{align}
The Laplacian acting on the function $C\!\raisebox{-0.4\height}{\six{3}{1}{1}{1}{1}{1}}$ does not produce any $-1$ values either.

For weight $w=9$, there are four combinations that do not produce any $-1$ edges after application of the Laplace operator. Besides the functions $C\!\raisebox{-0.4\height}{\six{4}{1}{1}{1}{1}{1}}$ and $C\!\raisebox{-0.4\height}{\six{2}{2}{2}{1}{1}{1}}$ they are
\begin{align}
( \Laplace + 18 ) \biggl( C_{\six{1}{1}{1}{2}{2}{2}}  + 3 C_{\six{1}{1}{2}{1}{2}{2}}  \biggr)= 72 C_{\three{1}{3}{5}} + 45 C_{\three{1}{4}{4}} + 36 C_{\three{2}{2}{5}} + 72 C_{\three{2}{3}{4}} - 36 E_4 E_5 + 36 E_9
\end{align}
and
\begin{align}
\Laplace \left(4 C_{\six{1}{1}{1}{1}{2}{3}}+C_{\six{1}{1}{2}{1}{3}{1}}\right) &=  -4 C_{\six{1}{1}{2}{2}{1}{2}} 
   +14 C_{\six{1}{1}{2}{1}{2}{2}} -16 C_{\six{1}{1}{1}{1}{2}{3}}+6 C_{\six{1}{1}{1}{2}{2}{2}}-16 C_{\five{3}{1}{3}{1}{1}}\nn\\
&\quad   -16 C_{\five{3}{1}{2}{2}{1}} -16 C_{\five{2}{1}{2}{3}{1}}-4 C_{\five{2}{2}{2}{2}{1}}-4 C_{\five{2}{1}{1}{2}{3}}\\
&\quad-12 E_4 C_{\three{1}{2}{2}}+12 E_4 C_{\three{1}{2}{4}}+84 C_{\three{1}{2}{6}}+128 C_{\three{1}{3}{5}}
   +96   C_{\three{1}{4}{4}}\nn\\
&\quad +48 C_{\three{2}{2}{5}}+34 C_{\three{2}{3}{4}}+22 C_{\three{3}{3}{3}} -48 E_4 E_5-20 E_3 E_6+68 E_9\,.\nn
\end{align}

For weight $w=10$, there are seven combinations that do not produce any $-1$. For weight $w=11$ there are $11$ and for weight $w=12$ there are $19$ such combinations.

\subsection*{Acknowledgements}
We are grateful to M.~Gaberdiel, J.~Gerken, S.~He, O.~Schlotterer and P.~Vanhove for useful discussions and correspondence. We would also like to thank O.~Schlotterer for useful comments on a first draft of this paper.

\newpage


\appendix

\section{Spectrum of the Laplacian on \texorpdfstring{$C_{(s,t,p)}$}{C(s,t,p)}}
\label{app:2loop}

We shall be interested in spectrum of the modular Laplacian on the function $C_{(s,t,p)}$ defined in~\eqref{eq:Cstp}. As in~\cite{DHoker:2015foa} we will introduce a generating function defined by
\begin{align}
\label{eq:Wexp1}
\mathcal{W}(t_1,t_2,t_3|\tau) = \sum_{s,t,p=1}^\infty t_1^{s-1} t_2^{t-1} t_3^{p-1} C_{(s,t,p)}(\tau)\,.
\end{align}
It follows from~\eqref{eq:LCabc} that the generating satisfies the equation
\begin{align}
\label{eq:diff2}
\left(\Delta - \mf{L}^2\right) \mathcal{W}  = \mathcal{R}\,,
\end{align}
with (using $\partial_i \equiv \partial_{t_i}$)
\begin{align}
\label{eq:L22}
\mf{L}^2 &= \mf{D}^2 + \mf{D} + (t_1^2+t_2^2+t_3^2-2t_1t_2-2t_2t_3-2t_3t_1) (\partial_1\partial_2 + \partial_2\partial_3 + \partial_3\partial_1)\,,\\
\label{eq:D2}
\mathfrak{D}&= \sum_{i=1}^3 t_i \partial_i\,,
\end{align}
and
\begin{align}
\label{eq:R2}
\mathcal{R}&=\sum_{s,t=0}^\infty \left( t_1^s t_2^t + t_2^s t_3^t + t_3^s t_1^t\right) \mathcal{R}_{st}^0 + \sum_{s,t=0}^\infty \left( t_1^s t_2^t t_3+ t_2^s t_3^t t_1+ t_3^s t_1^t t_2\right) \mathcal{R}_{st}^1 \,,\\
\mathcal{R}_{st}^0 &= 3s(t+1) E_{s+1}E_{t+2} + 3(s+1)t E_{s+2} E_{t+1} + (2-s-t-4st) E_{s+t+3}\,,\nn\\
\mathcal{R}_{st}^1 &= st \left( E_{s+2}E_{t+2}-E_{s+t+4}\right)\,.\nn
\end{align}
The `remainder' $\mathcal{R}$ is of lower complexity and represents some power series in the $t_i$ multiplying Eisenstein series or products of Eisenstein series. The spectral problem concerns the diagonalisation of the operator $\mf{L}^2$ in~\eqref{eq:L22}.

Everything in equation~\eqref{eq:diff2} is symmetric under the action of $\mathfrak{S}_3$, the symmetric group on three letters, acting on the $t_i$ in the fundamental representation.\footnote{Strictly speaking, the action of $\mathfrak{S}_3$ is originally on the Schwinger parameters $L_i$ in the fundamental representation and dually on the $t_i$. In this case, the two actions are the same.} Moreover, everything commutes with the weight operator $\mf{D}$ of~\eqref{eq:D2} that measures the degree of homogeneous polynomials in the $t_i$. Since $\mathcal{W}$ is symmetric in the $t_i$, only symmetric polynomials appear in the expansion on the right-hand side of~\eqref{eq:Wexp1}.

\subsection{\texorpdfstring{$SL(2)$}{SL(2)} Casimir in dual Schwinger space}

The following is a heuristic derivation of an exact rewriting of the differential operator $\mf{L}^2$. The vacuum two-loop diagram in cubic scalar field theory (a.k.a. sunset or melon graph) has the form (for unequal masses)
\begin{align}
\int d^D p_1 d^D p_2 \frac{1}{p_1^2+m_1^2}\frac{1}{p_2^2+m_2^2}\frac{1}{(p_1+p_2)^2+m_3^2}.
\end{align}
Using Schwinger parameters it can be related to\footnote{In this heuristic derivation, we are systematically ignoring factors of $2\pi$ etc. }
\begin{align}
\label{eq:2L}
\prod_{i=1}^3 \int_0^\infty dL_i \left( \det \Omega\right)^{-D/2} e^{-(L_1 m_1^2 +L_2 m_2^2 +L_3 m_3^2)}\,,
\end{align}
where 
\begin{align}
\Omega = \begin{pmatrix}
L_1+ L_3 & L_3\\
L_3 & L_2+ L_3
\end{pmatrix}\,.
\end{align}
This matrix of Schwinger parameters carries a natural action of $M\in PSL(2,\mathbb{R})$ by $\Omega \to M \Omega M^T$. One can even allow elements $M$ with determinant minus one here. This will happen for reflections below.

In order for the vacuum amplitude to be invariant under $PSL(2,\reals)$ one has to act correspondingly on the masses. We first rewrite this by defining $t_i = -m_i^2$ and then
\begin{align}
-(L_1 m_1^2 +L_2 m_2^2 +L_3 m_3^2) = \sum_{i=1}^3 L_i t_i = \mathrm{Tr} \, \Omega T
\end{align}
with 
\begin{align}
T = \begin{pmatrix}
t_1 & \frac12\left(t_3-t_1-t_2\right)\\
\frac12\left(t_3-t_1-t_2\right) & t_2
\end{pmatrix}\,.
\end{align}
The action of $M\in PSL(2,\mathbb{R})$ on this matrix is given by $T\to (M^{-1})^T T M^{-1}$. From this one can work out the following form of the infinitesimal generators of $PSL(2,\mathbb{R})$ in a Chevalley basis\footnote{This means that $[e,f]=h$, $[h,e]=2e$ and $[h,f]=-2f$.}
\begin{subequations}
\label{eq:diffops2}
\begin{align}
e &= (t_1 +t_2-t_3) \partial_1 + (t_1-t_2-t_3)\partial_3,&\\
f &= (t_1 +t_2-t_3) \partial_2 + (t_2-t_1-t_3)\partial_3,&\\
h &= -2t_1 \partial_1 +2t_2 \partial_2 -2(t_1-t_2)\partial_3.&
\end{align}
\end{subequations}
The quadratic Casimir is then seen to agree with~\eqref{eq:L22}
\begin{align}
\label{eq:Cas2}
\mathfrak{C}^2 &= \frac12ef+\frac12 fe + \frac 14h^2 \\
&= (t_i \partial_i)^2 + t_i \partial_i +(t_1^2+t_2^2+t_3^2-2t_1t_2-2t_2t_3-2t_3t_1) (\partial_1\partial_2 + \partial_2\partial_3 + \partial_3\partial_1)\nn\,.
\end{align}
Since $\mf{C}^2$ preserves by construction the degree of a polynomial, we can simultaneously diagonalise $\mf{C}^2$ and $\mf{D}$, while preserving the invariance under the symmetric group $\mf{S}_3$. The $\mf{S}_3$ invariant eigenfunctions of $\mf{D}$ are symmetric homogeneous polynomials. We note that even though the operator $\mathfrak{C}^2$ is the Casimir of $\mf{sl}(2)$, the individual operators~\eqref{eq:diffops2} do not act on the space of homogeneous \textit{symmetric} polynomials even though they act on homogeneous polynomials. This can be seen for example already for linear polynomials
\begin{align}
h (t_1+t_2+t_3) = 4(t_2-t_1)\,,
\end{align} 
which is not symmetric. Therefore the common $\mf{S}_3$ invariant eigenspaces of $\mf{D}$ are not representations of $\mf{sl}(2)$. 

\subsection{Spectrum using Molien's theorem}

We will nevertheless be able to exploit the representation theory of $\mf{sl}(2)$ to characterise the spectrum of $\mf{C}^2$. The reason for this is that the quadratic Casimir $\mf{C}^2$ preserves the space of symmetric polynomials as it is symmetric itself. Its possible eigenvalues are the ones inherited from the action on arbitrary homogeneous polynomials (that are a representation of $\mf{sl}(2)$). Denoting the $\mf{sl}(2)$ representation of dimension $p+1$ by the standard Dynkin label $[p]$, one has that on the irreducible representation $[p]$ the Casimir has the eigenvalue $\mf{C}^2 = \frac14 p(p+2)=s(s-1)$ for $s=\frac{p}2+1$. 

The space of linear homogeneous polynomials is the three-dimensional representation $[2]$ with basis $\{t_i\,|\, i=1,2,3\}$. The eigenvalue of $\mf{L}^2$ on this space is $2$, corresponding to $s=2$. Similarly, the homogeneous polynomials of degree $k$ are in the representation $S^k([2])$ of $\mf{sl}(2)$, where $S^k$ denotes the $k$th symmetric tensor power. (The symmetry is simply due to the fact that the $t_i$ commute.) The representation of $\mf{sl}(2)$ tells us that
\begin{align}
\label{eq:Spec21}
S^k([2]) = [2k] \oplus [2k-4] \oplus \ldots  \oplus [0/2]\,,
\end{align}
where the last term is meant to indicate $[0]$ or $[2]$ depending on whether $k$ is even or odd. This means that the spectrum of eigenvalues of $\mathfrak{L}^2$ for degree $k$ are given by
\begin{align}
\label{eq:Spec22}
s(s-1)\quad\textrm{with}\quad s=k+1,k-1,\ldots,1/2\,,
\end{align}
corresponding to all possible (bosonic) representations of dimensions equal to $2s-1$. The last term $1$ or $2$ again depends on the parity of $k$.

Having established the possible eigenvalues of $\mf{L}^2$, a harder question is to fix the degeneracies/multiplicities. For this we need to find the number of $\mf{S}_3$ singlets in a given representation $[p]$ of $\mf{sl}(2)$. A similar mathematical problem arose in a different context in~\cite{Gaberdiel:2014cha}. We here employ a different method based on Molien series.

Molien's theorem gives the number of invariants of a finite group (like $\mf{S}_3$) of fixed degree $k$ acting in a finite-dimensional representation of the group. We note that the standard representation of $\mf{S}_3$ is two-dimensional. In terms of Schwinger parameters it can be represented as
\begin{align}
L_1 &\leftrightarrow L_2  & M &=\begin{pmatrix}0&1\\1&0\end{pmatrix}\,,\nn\\
L_2 &\leftrightarrow L_3  & M &=\begin{pmatrix}1&-1\\0&-1\end{pmatrix}\,,\\
L_1 &\leftrightarrow L_3  & M &=\begin{pmatrix}-1&0\\-1&1\end{pmatrix}\,.\nn
\end{align}
and this embeds in the two-dimensional representation $[1]$ of $PGL(2,\ints)\subset PGL(2,\reals)$. The representation theory of $\mf{sl}(2)$ then allows us to determine
\begin{align}
S^p([1]) = [p]
\end{align}
such that this symmetric tensor product yields only a single representation. Therefore, the $\mf{S}_3$ singlets in the representation $[p]$ of $\mf{sl}(2)$ is the same as the degree $p$ invariants in the standard representation of $\mf{S}_3$. Molien's theorem then directly gives the generating function of the number singlets $n_p^{\mf{S}_3}$ in the representation $[p]$ as the coefficient of $q^p$ in 
\begin{align}
\sum_{p=0}^\infty n_p^{\mf{S}_3} q^p = \frac{1}{(1-q^2)(1-q^3)} = 1+q^2 + q^3+ q^4 + \ldots\,.
\end{align}
One can deduce the following closed formula for $n_{2k}^{\mf{S}_3}$ from this Molien series by expanding the geometric series:
\begin{align}
n_{2k}^{\mf{S}_3} = \left\lfloor\frac{k+3}{3}\right\rfloor.
\end{align}
We have restricted to even $p=2k$ since these are the only values that arise in the spectrum of $\mf{C}^2$ in view of~\eqref{eq:Spec21}. 

Combining this with~\eqref{eq:Spec22}, we deduce that on symmetric homogeneous polynomials of degree $k$ one has the following spectrum for $\mf{C}^2$: The eigenvalue $s(s-1)$ with multiplicity $\lfloor(s+2)/3\rfloor$ for the values $s=k+1,k-1,\ldots,1/2$. This is in complete agreement with Theorem~1 of~\cite{DHoker:2015foa} but without the need to explicitly diagonalise the operator.

If one is interested in finding an explicit set of eigenfunctions of $\mf{C}^2$ for a given eigenvalue and degree $k$, one can work in arbitrary basis of homogeneous symmetric polynomials, e.g. Schur polynomials. Since the eigenspaces can be degenerate, one could introduce an addition operator that commutes with $\mf{C}^2$ and $\mf{D}$ and that resolves the multiplicity. This is the approach taken in~\cite{DHoker:2015foa}. Alternatively, one could just introduce a random labelling of the various eigenfunctions in a given eigenspace since finding such an operator is not always obvious. Implementing the explicit diagonalisation at low degrees $k$ is straight-forward to implement on a computer and has been treated in detail in~\cite{DHoker:2015foa}.

\section{Heuristic for the \texorpdfstring{$\mf{sl}(3)$}{sl(3)} Casimir and the tetrahedral graph}
\label{app:SL3C}

We can use a similar heuristic to show that the differential operator appearing in~\eqref{eq:LG6} is closely related to the quadratic Casimir of $SL(3,\reals)$. For this we consider a cubic scalar vacuum diagram with tetrahedral topology shown in figure~\ref{fig:3loopH}.

\begin{figure}
\def\xshift{7}
\centering
\begin{tikzpicture}[very thick,decoration={
    markings,
    mark=at position 0.5 with {\arrow{>}}}
    ] 
  \draw[postaction=decorate] (0,0)--(0,2);
  \draw[postaction=decorate] (0,0)--(1.73,-1);  
  \draw[postaction=decorate] (-1.73,-1)--(0,0);
  \draw[postaction=decorate] (-1.73,-1)--(0,2);
  \draw[postaction=decorate] (0,0)--(0,2);
  \draw[postaction=decorate] (1.73,-1)--(-1.73,-1);
  \draw[postaction=decorate] (0,2)--(1.73,-1);  
  \draw (-1.2,0.8) node {$p_1$};
  \draw (-0.3,0.7) node {$p_2$};  
  \draw (1.6,0.7) node {$p_1+p_2$};  
  \draw (0.7,-0.1) node {$p_3$};  
  \draw (0,-1.3) node {$p_1+p_2+p_3$};    
  \draw (-0.2,-0.6) node {$p_2+p_3$};    
   \draw[postaction=decorate] (0+\xshift,0)--(0+\xshift,2);
  \draw[postaction=decorate] (0+\xshift,0)--(1.73+\xshift,-1);  
  \draw[postaction=decorate] (-1.73+\xshift,-1)--(0+\xshift,0);
  \draw[postaction=decorate] (-1.73+\xshift,-1)--(0+\xshift,2);
  \draw[postaction=decorate] (0+\xshift,0)--(0+\xshift,2);
  \draw[postaction=decorate] (1.73+\xshift,-1)--(-1.73+\xshift,-1);
  \draw[postaction=decorate] (0+\xshift,2)--(1.73+\xshift,-1);  
  \draw (-1.2+\xshift,0.8) node {$t_1$};
  \draw (-0.3+\xshift,0.7) node {$t_2$};  
  \draw (1.3+\xshift,0.7) node {$t_3$};  
  \draw (0.7+\xshift,-0.1) node {$t_4$};  
  \draw (0+\xshift,-1.3) node {$t_5$};    
  \draw (-0.5+\xshift,-0.6) node {$t_6$};     
\end{tikzpicture}
\caption{\label{fig:3loopH}\sl The tetrahedral graph with labelling of momenta and parameters of the generating function.}
\end{figure}
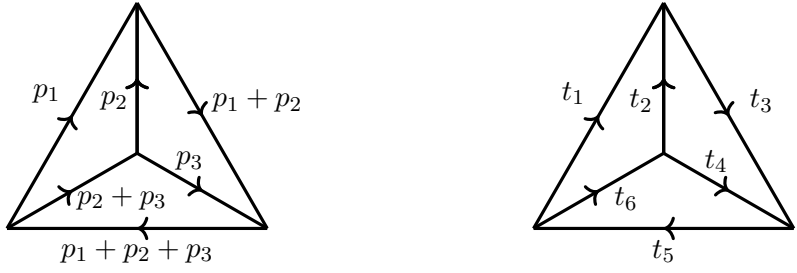

One obtains an expression similar to~\eqref{eq:2L} in terms of 
\begin{align}
\prod_{i=1}^6 \int_0^\infty dL_i (\det \Omega)^{-D/2} e^{\mathrm{Tr}\, \Omega T}
\end{align}
where now
\begin{align}
\Omega &= \begin{pmatrix}
L_1 +L_3 + L_5 & L_3+L_5 & L_5\\
L_3 + L_5 & L_2+L_3+L_5+L_6 & L_5+L_6\\
L_5 & L_5+L_6 &L_4+L_5+L_6
\end{pmatrix}\,,\nn\\
T &=\begin{pmatrix}
t_1 & \frac12(t_3-t_1-t_2) & \frac12(t_2-t_3-t_5-t_6)\\
\frac12(t_3-t_1-t_2) & t_2 & \frac12(t_6-t_2-t_4)\\
\frac12(t_2-t_3-t_5-t_6) & \frac12(t_6-t_2-t_4) & t_4
\end{pmatrix}\,.
\end{align}
The action is now by $SL(3,\mathbb{R})$. The quadratic Casimir in the $t_i$ variables becomes
\begin{align}
\label{eq:CasSL3}
\mathfrak{C}^2 &= \frac43 \mathfrak{D}^2 + 2 \mathfrak{D} +
\sum_{V_{ijk}} (t_i^2+t_j^2+t_k^2 - 2t_i t_j -2 t_j t_k -2t_k t_i) (\partial_i \partial_j +\partial_j \partial_k +\partial_k\partial_i)\nn\\
& \quad\quad\quad + \sum_{i=1}^3 ((t_{i+1} - t_{i+2} -t_{i+4} +t_{i+5})^2 -4 t_i t_{i+3}) \partial_i \partial_{i+3}\,.
\end{align}
The meaning of the various terms here is as follows. The scaling operator
\begin{align}
\mathfrak{D} = \sum_{i=1}^6 t_i \partial_i
\end{align}
measures the degree of homogeneous polynomials in the $t_i$. The first sum is over the four vertices $V_{ijk}$ of the tetrahedral graph, so $(ijk)\in\{ 123,\,135,\, 156,\,246\}$, and the sum contains all second derivatives of adjacent edges. The last term contains mixed second derivatives over opposite (non-adjacent) edges and there are three such pairs. If an index exceeds $6$ it is to be read modulo $6$. Up to $\mf{D}$ terms in~\eqref{eq:CasSL3}, we recognise the same differential operator as the one appearing in~\eqref{eq:LG6}.

\section{Graphical derivation of tetrahedral Laplace equation}
\label{app:LC6}

In order to evaluate the modular Laplacian on the function $C\!\!\raisebox{-0.3\height}{\six{s}{t}{p}{q}{r}{w}}$ we work with the deformation calculus of~\eqref{eq:Ldef}. This means that we have to distribute the deformation differentials $\delta_\mu$ and $\bar{\delta}_\mu$ on the lines of the tetrahedral diagram. In this appendix we draw the tetrahedral graph as a Mercedes diagram in order to unclutter some of the equations. The reference graph is
\begin{align}
\merc{}{}{}{}{}{}{}{}{}{}{}{}{}{}{}{}{}{} 
 \equiv \merc{s}{t}{p}{q}{r}{w}{}{}{}{}{}{}{}{}{}{}{}{} 
\end{align}
and in all following equations we will only put changes relative to the values $s$, $t$, $p$, $q$, $r$ and $w$ on the diagram.

There are a number of different possibilities when placing $\delta_\mu$ and $\bar\delta_\mu$ on the diagram. They can be placed either  $(i)$ on the same line,  $(ii)$ on adjacent lines or $(iii)$ on opposite lines.

Case $(i)$ is simplest:
\begin{align}
s(s-1) \merc{\delta_\mu\bar\delta_\mu}{}{}{}{}{}{}{}{}{}{}{}{}{}{}{}{}{} = s(s-1) \merc{}{}{}{}{}{}{}{}{}{}{}{}{}{}{}{}{}{}\,,
\end{align}
where we have included the relevant combinatorial factor for putting $\delta_\mu$ and $\bar\delta_\mu$ on two different of the $s$ many propagators of this line. Expanding first $\delta_\mu$ and $\bar\delta_\mu$ into additional propagators and partial world-sheet derivatives $\partial$ and $\bar{\partial}$ according to~\eqref{eq:deltaGDiag} and then contracting adjacent $\partial$ and $\bar{\partial}$ using~\eqref{eq:ddcontract} immediately gives back the original diagram. (The second term in the contraction rule~\eqref{eq:ddcontract} never contributes in the considerations of this appendix as it always gives tadpole diagrams that vanish thanks to~\eqref{eq:tadpole}.) Thus this part of the action of $\Delta=\delta_\mu\bar{\delta}_\mu$ contributes to the `eigenvalue' part of the differential equation. There are naturally similar terms for all the other five lines.

Case $(ii)$ is slightly more involved. There are $12$ pairs of adjacent lines (three per vertex) and they all have similar contributions. We consider the 
\begin{align}
st \left[ \merc{\delta_\mu}{\bar\delta_\mu}{}{}{}{}{}{}{}{}{}{}{}{}{}{}{}{} +c.c. \right]\,,
\end{align}
where we have noted that one always has to add the complex conjugate with $\delta_\mu$ and $\bar\delta_\mu$ interchanged. The diagram shown can be manipulated as follows\footnote{In this and the following equations we do not write out the factors of $\pi$ and $\tau_2$ as they cancel in the final expression.}
\begin{align}
\merc{\delta_\mu}{\bar\delta_\mu}{}{}{}{}{}{}{}{}{}{}{}{}{}{}{}{}  &= \merc{+1}{+1}{}{}{}{}{\textcolor{blue}{\pa}}{}{\pab}{}{}{}{\pa}{}{}{\pab}{}{}
=\merc{+1}{}{}{}{}{}{}{}{}{}{}{}{\pa}{}{}{\pab}{}{} -\merc{+1}{+1}{}{}{}{}{}{\pa}{\textcolor{blue}{\pab}}{}{}{}{\pa}{}{}{\pab}{}{}\nn\\
&= \merc{+1}{}{}{}{}{}{}{}{}{}{}{}{\pa}{}{}{\pab}{}{} -\merc{+1}{+1}{-1}{}{}{}{}{}{}{}{}{}{\pa}{}{}{\pab}{}{}+\merc{}{+1}{}{}{}{}{}{}{}{}{}{}{\pa}{}{}{\pab}{}{}
\end{align}
In this equation we have shown in blue in each step the derivatives that are integrated by parts at the trivalent vertices. At this point we can apply Lemma~\ref{lemma:adj} below to all three diagrams  to get
\begin{align}
\merc{\delta_\mu}{\bar\delta_\mu}{}{}{}{}{}{}{}{}{}{}{}{}{}{}{}{} +c.c. &=
 \merc{+1}{-1}{}{}{}{}{}{}{}{}{}{}{}{}{}{}{}{}
 +\merc{-1}{+1}{}{}{}{}{}{}{}{}{}{}{}{}{}{}{}{}
 +\merc{+1}{+1}{-2}{}{}{}{}{}{}{}{}{}{}{}{}{}{}{}\nn\\
&\quad\quad -2\merc{+1}{}{-1}{}{}{}{}{}{}{}{}{}{}{}{}{}{}{}
 -2\merc{}{+1}{-1}{}{}{}{}{}{}{}{}{}{}{}{}{}{}{}
\end{align}
This kind of manipulation is sufficient to find the Laplace equation for $C\!\!\raisebox{-0.3\height}{\three{s}{t}{p}}$ and concludes case $(ii)$.

For case $(iii)$ we have to put the differentials on opposite lines, for example
\begin{align}
sq \left[ \merc{\delta_\mu}{}{}{\bar\delta_\mu}{}{}{}{}{}{}{}{}{}{}{}{}{}{} +c.c. \right]\,.
\end{align}
We start manipulating the diagram with the aim of reducing it to diagrams with one holomorphic and one anti-holomorphic world-sheet derivatives. All such diagrams can be simplified using the lemmas below.
\begin{align}
\merc{\delta_\mu}{}{}{\bar\delta_\mu}{}{}{}{}{}{}{}{}{}{}{}{}{}{} &= \merc{+1}{}{}{+1}{}{}{\pa}{}{}{}{}{\pab}{\pa}{}{}{}{\textcolor{blue}{\pab}}{}
 = \merc{+1}{}{}{+1}{}{}{\textcolor{blue}{\pa}}{}{}{}{}{}{\pa}{}{}{\pab}{\pab}{} + \merc{+1}{}{}{+1}{}{}{\pa}{}{}{}{}{}{\textcolor{blue}{\pa}}{}{}{}{\pab}{\pab} \nn\\
&= -\merc{+1}{-1}{}{+1}{}{}{}{}{}{}{}{}{\pa}{}{}{}{\pab}{}    +\merc{+1}{}{}{+1}{}{}{\pa}{\pa}{}{}{}{}{}{}{}{\textcolor{blue}{\pab}}{\pab}{} 
     -\merc{+1}{}{}{+1}{}{-1}{\pa}{}{}{}{}{}{}{}{}{}{\pab}{} + \merc{+1}{}{}{+1}{}{}{}{}{}{}{}{}{\pa}{}{\pa}{}{\pab}{\textcolor{blue}{\pab}}
\end{align}
The two terms with the minus sign can be treated with Lemma~\ref{lemma:opp} below. The other two terms both have vertices with three derivative sitting on them after moving the blue ones to the other end of the line. Integrating by parts then the derivatives in blue reduces them to terms to which one can also apply the lemmas below and one requires both. Writing out all terms gives some cancellations and the total for case $(iii)$ becomes finally:
\begin{align}
\merc{\delta_\mu}{}{}{\bar\delta_\mu}{}{}{}{}{}{}{}{}{}{}{}{}{}{}\!\!\! +c.c. &= - 2 \merc{}{}{}{}{}{}{}{}{}{}{}{}{}{}{}{}{}{} 
   +2 \merc{+1}{}{-1}{+1}{}{-1}{}{}{}{}{}{}{}{}{}{}{}{}  +2 \merc{+1}{-1}{}{+1}{-1}{}{}{}{}{}{}{}{}{}{}{}{}{} \nn\\
&\quad +\merc{+1}{-2}{}{+1}{}{}{}{}{}{}{}{}{}{}{}{}{}{} +\merc{+1}{}{-2}{+1}{}{}{}{}{}{}{}{}{}{}{}{}{}{}  +\merc{+1}{}{}{+1}{-2}{}{}{}{}{}{}{}{}{}{}{}{}{}  +\merc{+1}{}{}{+1}{}{-2}{}{}{}{}{}{}{}{}{}{}{}{}\nn\\
&\quad -2 \merc{+1}{}{}{+1}{-1}{-1}{}{}{}{}{}{}{}{}{}{}{}{} -2 \merc{+1}{-1}{}{+1}{}{-1}{}{}{}{}{}{}{}{}{}{}{}{} -2 \merc{+1}{}{-1}{+1}{-1}{}{}{}{}{}{}{}{}{}{}{}{}{} -2 \merc{+1}{-1}{-1}{+1}{}{}{}{}{}{}{}{}{}{}{}{}{}{}\!.
\end{align}
One sees that there is a contribution to the eigenvalue from case $(iii)$; the remaining terms have been grouped according to whether they use opposite or adjacent lines in addition to the lines with the differentials. This concludes case $(iii)$.

Putting all the cases together gives~\eqref{eq:LC6}.

\subsection{Two lemmas on first derivative graphs}

We give two simple lemmas for tetrahedral diagrams that have one derivative $\pa$ and one derivative $\pab$ on them.

\begin{lemma}[Adjacent lines lemma]
\label{lemma:adj}
If the derivatives are on adjacent lines one has
\begin{align}
\merc{}{}{}{}{}{}{\pa}{}{\pab}{}{}{}{}{}{}{}{}{}  +c.c. = 
\merc{-1}{}{}{}{}{}{}{}{}{}{}{}{}{}{}{}{}{}  + \merc{}{-1}{}{}{}{}{}{}{}{}{}{}{}{}{}{}{}{} -\merc{}{}{-1}{}{}{}{}{}{}{}{}{}{}{}{}{}{}{} \,,
\end{align}
such that the lines with the derivatives appear with the same sign and the last line at the vertex with an opposite sign.
\end{lemma}

\begin{proof}
The proof is by direct calculation:
\begin{align}
\merc{}{}{}{}{}{}{\pa}{}{\pab}{}{}{}{}{}{}{}{}{} &= 
  \merc{-1}{}{}{}{}{}{}{}{}{}{}{}{}{}{}{}{}{} - \merc{}{}{}{}{}{}{\pa}{\pab}{}{}{}{}{}{}{}{}{}{}\nn\\  
&= \merc{-1}{}{}{}{}{}{}{}{}{}{}{}{}{}{}{}{}{} -\merc{}{}{-1}{}{}{}{}{}{}{}{}{}{}{}{}{}{}{}  + \merc{}{}{}{}{}{}{}{\pab}{\pa}{}{}{}{}{}{}{}{}{}\nn\\  
&= \merc{-1}{}{}{}{}{}{}{}{}{}{}{}{}{}{}{}{}{} -\merc{}{}{-1}{}{}{}{}{}{}{}{}{}{}{}{}{}{}{}  + \merc{}{-1}{}{}{}{}{}{}{}{}{}{}{}{}{}{}{}{} - \merc{}{}{}{}{}{}{\pab}{}{\pa}{}{}{}{}{}{}{}{}{}\,.
\end{align}
As the last term in the last line is the complex conjugate of the original diagram, the assertion follows.
\end{proof}

\begin{lemma}[Opposite lines lemma]
\label{lemma:opp}
If the derivatives are on opposite (non-adjacent) lines one has
\begin{align}
\merc{}{}{}{}{}{}{}{\pab}{}{}{}{}{}{\pa}{}{}{}{}  +c.c. = 
\merc{-1}{}{}{}{}{}{}{}{}{}{}{}{}{}{}{}{}{}  + \merc{}{}{}{-1}{}{}{}{}{}{}{}{}{}{}{}{}{}{} -\merc{}{-1}{}{}{}{}{}{}{}{}{}{}{}{}{}{}{}{} -\merc{}{}{}{}{-1}{}{}{}{}{}{}{}{}{}{}{}{}{}  \,,
\end{align}
such that all lines without derivatives are affected; opposite lines appear with the same sign.
\end{lemma}

\begin{proof}
We calculate
\begin{align}
\merc{}{}{}{}{}{}{}{\pab}{}{}{}{}{}{\pa}{}{}{}{}  +c.c. = \merc{}{}{}{}{}{}{\pa}{\pab}{}{}{}{}{}{}{}{}{}{} + \merc{}{}{}{}{}{}{}{\pab}{}{}{\pa}{}{}{}{}{}{}{}  + c.c.
\end{align}
Both diagrams are now such that the derivatives are on adjacent lines and one can apply Lemma~\ref{lemma:adj} to each of the two diagrams, leading to six diagrams out of which two cancel. The remaining four are the asserted ones.
\end{proof}

%
%


\end{document}